\documentclass[11pt,cls,onecolumn]{IEEEtran}

\usepackage{subfigure,cite,graphicx,amsmath,amssymb,eufrak,mathrsfs,epsfig,dblfloatfix}
\usepackage{amsmath,amssymb,mathrsfs,graphicx}
\usepackage{psfrag}

\DeclareMathOperator\C{\sf C} \DeclareMathOperator\E{\sf E}

\newtheorem{definition}{Definition}
\newtheorem{theorem}{Theorem}
\newtheorem{example}{Example}
\newtheorem{remark}{Remark}
\newtheorem{lemma}{Lemma}

\newtheorem{corollary}{Corollary}

\begin{document}
\title{Computation Over Gaussian Networks With Orthogonal Components}

\author{Sang-Woon Jeon\IEEEmembership{, Member, IEEE}, Chien-Yi Wang\IEEEmembership{, Student Member, IEEE}, and\\ Michael Gastpar\IEEEmembership{, Member, IEEE}
\thanks{This work has been supported in part by the European ERC Starting Grant 259530-ComCom. The first author was also funded in part by the MSIP (Ministry of Science, ICT \& Future Planning), Korea in the ICT R \& D Program 2013.}
\thanks{The material in this paper was presented in part at the Information Theory and Applications Workshop (ITA), San Diego, CA, February 2013 and the IEEE International Symposium on Information Theory (ISIT), Turkey, Istanbul, July 2013.}
\thanks{S.-W. Jeon is with the Department of Information and Communication Engineering, Andong National University, South Korea (e-mail: swjeon@anu.ac.kr).}

\thanks{C.-Y. Wang is with the School of Computer and Communication Sciences, Ecole Polytechnique F{\'e}d{\'e}rale de Lausanne (EPFL), Lausanne,
Switzerland (e-mail: chien-yi.wang@epfl.ch).}%
\thanks{M. Gastpar is with the School of Computer and Communication Sciences, Ecole Polytechnique F{\'e}d{\'e}rale de Lausanne (EPFL), Lausanne,
Switzerland and the Department of Electrical Engineering and Computer Sciences, University of California, Berkeley, CA, USA (e-mail: michael.gastpar@epfl.ch).}
}
\maketitle

\IEEEpeerreviewmaketitle

\begin{abstract}
Function computation of arbitrarily correlated discrete sources over Gaussian networks with orthogonal components is studied.
Two classes of functions are considered: the arithmetic sum function and the type function. 
The arithmetic sum function in this paper is defined as a set of multiple weighted arithmetic sums, which includes averaging of the sources and estimating each of the sources as special cases.
The type or frequency histogram function counts the number of occurrences of each argument, which yields many important statistics such as mean, variance, maximum, minimum, median, and so on.
The proposed computation coding first abstracts Gaussian networks into the corresponding modulo sum multiple-access channels via nested lattice codes and linear network coding and then computes the desired function by using linear Slepian--Wolf source coding.
For orthogonal Gaussian networks (with no broadcast and multiple-access components), the computation capacity is characterized for a class of networks.
For Gaussian networks with multiple-access components (but no broadcast), an approximate computation capacity is characterized for a class of networks.
\end{abstract}


\section{Introduction} \label{sec:intro}
In wireless sensor networks, the goal of communication is typically for a fusion center to learn a {\em function}
of the sensor observations, rather than the raw observations themselves. Examples include distributed averaging,
alarm detection, environmental monitoring, and so on. The fundamental paradigm of digital communication
suggests that each sensor should independently compress its observations (using sophisticated compression
techniques, taking into account possible correlations in the observations as well as the fact that the
fusion center is only interested in a function of the observations), whereupon these compressed versions
are communicated reliably (at negligible error probability) to the fusion center. For point-to-point communication,
this architecture has been shown to be optimal by Shannon~\cite{Shannon:48}, a result that is sometimes referred to
as the {\em source--channel separation theorem.} For general networked communication, however, it is well known
that this digital communication paradigm leads to suboptimal performance, see e.g.~\cite{Cover:80}.
Furthermore, in terms of the number of nodes in the network, the suboptimality can be dramatic~\cite{Gastpar:08}.
Hence, for the communication problem where a fusion center needs to learn a function of the sensor observations,
it is beneficial to consider {\em joint source--channel coding.}

Communication strategies for the problem of function computation over networks have been actively studied in
the literature, see e.g. \cite{Korner:79,Ahlswede:83,Nazer:07,Nazer:11,Zhan:13} and the reference therein.
For one class of strategies of function computation over wireless networks, the essence is to exploit the superposition
property of wireless channels to more efficiently compute the desired function.
Roughly speaking, previous work in this area can be categorized into three classes: the modulo-$p$ sum computation over (noisy) modulo-$p$ sum networks \cite{Korner:79,Ahlswede:83,Nazer:07,Suh:12,Zhan:13}, the modulo-$p$ sum computation over Gaussian networks assuming an arbitrarily large $p$ \cite{Nazer:11,Zhan:13}, and the sum of Gaussian sources over Gaussian networks under the mean squared error distortion  \cite{Soundararajan:12,Zhan:13}.       
All these works rely on joint source--channel coding in order to exploit the \emph{similarity between sum-type functions and the superposition property of wireless channels}.
We can easily find examples that this joint source--channel coding approach significantly improves an achievable computation rate or decrease an achievable  distortion compared to the source--channel separation approach.

In spite of the previous work, however, it is still unclear how to efficiently compute fundamental sample statistics such as sample mean, variance, maximum, minimum, and so on  over Gaussian networks.
As mentioned before, many sensor applications are interested in the sample mean, for instance, average temperature from several temperature readings.
For alarm detection, a relevant function will be the maximum or minimum value among the measurements.
One naive approach is to estimate each of the measurements separately, which is universal in the sense that any function of the measurements can be deduced accordingly. 
Unfortunately, it turns out that this naive approach is quite suboptimal in terms of computation rate for most functions of interest.
Another extreme approach is to tackle each function case by case, but we may want to avoid this approach too since there exist numerous important functions to be considered. 
Therefore, it would be nice to come up with a  general coding scheme that is able to compute a broad class of functions including the above fundamental functions but at the same time provide a better computation rate than the separation-based computation.

To achieve this goal, we focus on computing \emph{the type or frequency histogram function} in this paper.
For a better understanding, consider the type computation over the Gaussian multiple-access channel (MAC) depicted in Fig. \ref{figs:averaging}.
The  $K$ sensors observe their discrete sources $S_1,S_2,\cdots,S_K\in\{0,1,\cdots, p-1\}$, which can be arbitrarily correlated to each other, and the fusion center wishes to reliably compute its type, represented as 
\begin{equation} \label{eq:type_sum}
\left(\sum_{i=1}^K\mathbf{1}_{S_i=0},\sum_{i=1}^K\mathbf{1}_{S_i=1},\cdots,\sum_{i=1}^K\mathbf{1}_{S_i=p-1}\right),
\end{equation}
where $\mathbf{1}_{(\cdot)}$ denotes the indicator function of an event.
As pointed out in \cite{Giridhar:05}, computing the type function is very powerful since it yields many important statistics such as sample mean, maximum, minimum, variance, median, mode, and so on. 
Basically, \emph{any symmetric function whose function value is invariant with respect to permutations of its arguments is computable from the type function}.
As seen in \eqref{eq:type_sum}, the type function consists of multiple arithmetic sums of indicator functions, which can be regarded as binary sources.
Therefore, a fundamental question for the type computation  is how to exploit the similarity between the arithmetic sum of discrete sources and the superposition property of real-valued transmit signals corrupted by additive noise.

\begin{figure}[t!]
\begin{center}
\includegraphics[scale=1.3]{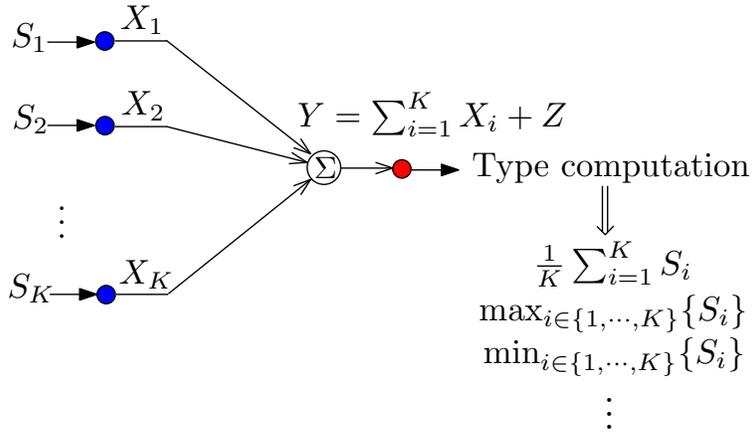}
\end{center}
\vspace{-0.15in}
\caption{Type computation over the Gaussian MAC.}
\label{figs:averaging}
\vspace{-0.1in}
\end{figure}

In this paper, we consider the computation of a more general class of functions over a general Gaussian network assuming some orthogonal components, which includes the problem in Fig. \ref{figs:averaging} as a special case.
Two classes of desired functions are considered: the arithmetic sum function and the type function. 
The former in this paper is defined as a set of multiple weighted arithmetic sums, which includes averaging of the sources and estimating each of the sources as special cases.
The latter is counting the number of occurrences of each argument among the sources, see \eqref{eq:type_sum}.
Regarding the channel model, we consider two types of Gaussian channels.
The first model is orthogonal Gaussian networks in which there is no broadcast and multiple-access component, which is equivalent to bit-pipe wired networks \cite{Ahlswede:00,Li:03,Koetter:03}.
The second model is Gaussian networks with multiple-access components (and no broadcast component), which includes Gaussian MACs, more generally Gaussian tree networks as special cases \cite{Nam:11,Zhan:13}.

\subsection{Contribution}
The main contributions of the paper is as follows.
\begin{itemize}
\item For orthogonal Gaussian single-hop networks, we propose a general computation code which includes both  Slepian--Wolf source coding and  K\"{o}rner--Marton linear source coding for computing. An example is presented to demonstrate the benefit of introducing K\"{o}rner--Marton linear source coding for computing when the sources are correlated. 

\item We extend  K\"{o}rner--Marton linear source coding for computing to general orthogonal Gaussian networks incorporated with linear network coding at each relay node. 
We characterize the computation capacity for a class of networks. 
The result demonstrates that, even without multiple-access component, K\"{o}rner--Marton linear source coding for computing is still beneficial for a broad class of  relay networks.
 
\item For Gaussian MACs, we propose a  computation code that first abstracts the original Gaussian MAC into the corresponding modulo sum channel via lattice codes and then applies K\"{o}rner--Marton linear source coding for computing on top of the transformed channel.
We show that the proposed computation code provides a much better computation rate than the separation-based computation, especially when the number of sources becomes large.


\item We extend the proposed computation code for Gaussian MACs to general Gaussian networks with multiple-access components. For this, we establish a general transformation method from Gaussian networks with multiple-access components into the corresponding modulo sum channels. On top of this transformed network, we apply the computation code proposed for Gaussian MACs. 
For a class of networks, we characterize an approximate computation capacity that provides a bounded gap from computation capacity, independent of power $P$. 
\end{itemize}

\subsection{Related Work}
In his seminal work \cite{Shannon:48}, Shannon showed that separation of source and channel coding is optimal for discrete memoryless point-to-point  channels.
However, source--channel separation is not optimal for general networks, for instance, the problems of sending correlated sources over MACs \cite{Lapidoth:10,Lim:10} or broadcast channels (BCs) \cite{Lapidoth2:10,Tian:11}. 
That is, joint source--channel coding is essentially required for sending correlated sources over networks.
Furthermore, it has been proved that an uncoded transmission scheme, a simple way of joint source-channel coding, is optimal or near-optimal for estimating  a source from several correlated observations over Gaussian networks \cite{Gastpar:05,Gastpar:08}.   

Function computation has been actively studied in the source coding perspective \cite{Korner:79,Ahlswede:83,Han:87,Orlitsky:01,Cuff:09,Ma:11,Huang:12,Huang2:12}.
In particular, computing the modulo-two sum has been considered in \cite{Korner:79} under the distributed source coding framework, which captures the potential of linear source coding \cite{Csiszar:82} for function computation.
A more general achievability has been proposed for the modulo-two sum computation in \cite{Ahlswede:83}  and for a general discrete function in \cite{Huang:12,Huang2:12}.
In \cite{Orlitsky:01}, computing a general function with the help of side information has been studied. 
Function computation has been also considered in the context of cascade source coding \cite{Cuff:09} and interactive source coding \cite{Ma:11}. 

The modulo sum or more generally linear function computation has been recently extended to relay networks under various channel models such as bit-pipe wired networks \cite{Appuswamy:11,Appuswamy3:11,Kowshik:12}, linear finite field networks \cite{Zhan:13}, Gaussian networks assuming no broadcast component \cite{Zhan:13} by incorporating linear network coding \cite{Ahlswede:00,Li:03,Koetter:03,Ho:06,Avestimehr:11} at each relay node.
Function multicasting has been studied for linear finite field interference channels \cite{Suh:12} and for undirected graphs \cite{Kannan:13,Kannan2:13}.
A more general classes of function computation over bit-pipe wired networks has been considered in \cite{Appuswamy:11,Appuswamy2:11}.

Scaling laws on function computation has been studied based on the collision model \cite{Giridhar:05,Ma:12},  in which  concurrent transmission from multiple senders is assumed to cause a collision and, therefore, is not allowable. In particular, it has been shown that the order of $\frac{1}{\log K}$ scaling law is achievable as the number of sources $K$ increases for the type or frequency histogram computation over collocated collision networks \cite{Giridhar:05,Ma:12}.
Recently, it has been shown that non-vanishing scaling law is achievable for the type-threshold function computation over collocated Gaussian network even as $K$ tends to infinity \cite{Wang:13}. 

The potential of linear source coding has been also captured by Nazer and Gastpar in \cite{Nazer:07}, applying the linear source coding in \cite{Korner:79} for the function computation over MACs.
An efficient way of computing the modulo sum or the sum of Gaussian sources over Gaussian MACs is to apply lattice codes \cite{Nazer:07,Soundararajan:12}, see also \cite{Erez:04,Erez:05,Nam:10} for lattice code construction. 
Lattice-based network computation has been recently extended to multiple receivers called compute-and-forward \cite{Nazer:11} in which each relay computes or decodes linear combination of the sources.
In \cite{Zhan:13}, a similar lattice code construction has been used for computing a linear function over linear finite field networks and the sum of Gaussian sources over Gaussian networks.

\begin{figure}[t!]
\begin{center}
\includegraphics[scale=1.3]{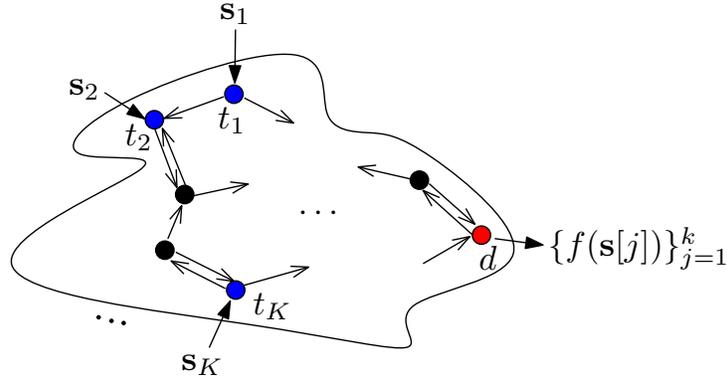}
\end{center}
\vspace{-0.15in}
\caption{Computation over a network in which node $t_i$ observes the length-$k$ source $\mathbf{s}_i=[s_i[1],\cdots,s_i[k]]^T$ and node $d$ wishes to compute the desired function $\{f(\mathbf{s}[j])\}_{j=1}^k$, where $\mathbf{s}[j]=[s_1[j],\cdots,s_K[j]]^T$.}
\label{figs:network}
\vspace{-0.1in}
\end{figure}

\section{Problem Formulation} \label{sec:prob}
Throughout the paper, we denote $[1:n]:=\{1,2,\cdots,n\}$, ${\sf C}(x):=\frac{1}{2}\log (1+x)$, and ${\sf C}^+(x):=\max\left\{\frac{1}{2}\log (x),0\right\}$.
For $x_i\in\mathbb{F}_p$, $\bigoplus_{i=1}^n x_i$ denotes the modulo-$p$ sum of $\{x_i\}_{i\in[1:n]}$, where $p$ is assumed to be a prime number.
Let  $\mathbf{1}_{(\cdot)}$ denote the indicator function of an event. For random variables $A$ and $B$, $H(A)$ denotes the entropy of $A$ and $I(A;B)$ denotes the mutual information between $A$.

\subsection{Network Model}
Consider a network represented by a directed graph $G=(V,E)$ depicted in Fig. \ref{figs:network}.
Denote the set of incoming and outgoing nodes at node $v\in V$ by $\Gamma_{in}(v)=\{u\in V:(u,v)\in E\}$ and $\Gamma_{out}(v)=\{u\in V:(v,u)\in E\}$, respectively.
Denote the $i$th sender, $i\in[1:K]$, by $t_i\in V$ and suppose that it observes a length-$k$ discrete source vector $\mathbf{s}_i=[s_i[1],\cdots,s_i[k]]^T\in[0:p-1]^k$.
Denote the set of $K$ sources at time $j$ by $\mathbf{s}[j]=[s_1[j],\cdots,s_K[j]]^T$.
The receiver $d\in V$ wishes to compute a symbol-by-symbol function of $K$ sources, i.e., $f(\mathbf{s}[j])$ for all $j\in[1:k]$.
We assume that $d\notin \{t_i\}_{i\in[1:K]}$ and $G$ contains a directed path from all nodes in $V$ to the receiver $d$.
Without loss of generality, we assume that the nodes with no incoming edge are included in $\{t_i\}_{i\in[1:K]}$ .

We mainly consider two desired functions: the arithmetic sum function and the type or frequency histogram function, whose formal definitions are given below. 

\begin{definition}[Arithmetic Sum Function]  \label{def:desired_function1}
Let $\mathbf{s}=[s_1,\cdots,s_K]^T\in[0:p-1]^{K}$. For the arithmetic sum computation, the desired function is given by  $f(\mathbf{s})=\{\sum_{i=1}^K a_{li}s_i\}_{l=1}^L$, where $a_{li}\in[0:p-1]$. Hence $f(\mathbf{s})\in [0:(p-1)^2K]^L$ for the arithmetic sum function.
\end{definition}

\begin{definition}[Type Function]  \label{def:desired_function2}
Let $\mathbf{s}=[s_1,\cdots,s_K]^T\in[0:p-1]^{K}$ and $b_l(\mathbf{s})=\sum_{i=1}^K\mathbf{1}_{s_i=l}$ for $l\in[0:p-1]$.
For the type computation, the desired function is given by  $f(\mathbf{s})=\{b_0(\mathbf{s}),\cdots,b_{p-1}(\mathbf{s})\}$. 
Hence $f(\mathbf{s})\in [0:K]^p$ for the type function.
\end{definition}

\begin{remark}[Arithmetic Sum and Type]\label{re:sum_type}
The arithmetic sum function in this paper is defined as multiple weighted arithmetic sums, which includes averaging of the sources and estimating each of the sources as special cases.  
As shown in Definition \ref{def:desired_function2}, the type function can be also represented as multiple arithmetic sums.
Therefore, the essence of the type computation is how to efficiently compute arbitrarily correlated multiple arithmetic sums over Gaussian networks.
\end{remark}

\begin{remark}[Symmetric Function Computation]
As pointed out by \cite{Giridhar:05}, computing the type function is very powerful since it yields many important statistics such as sample mean, maximum, minimum, variance, median, mode, and so on. 
Basically, any symmetric function whose function value is invariant with respect to permutations of its arguments is computable from the type function.
That is, symmetric functions satisfy  
$f(s_1,s_2,\cdots,s_K)=f(s_{\sigma_1},s_{\sigma_2},\cdots,s_{\sigma_K})$
for any permutation set $\{\sigma_i\}_{i\in[1:K]}$ and, therefore, they are deterministic functions of the type function.
\end{remark}

We assume arbitrarily correlated stationary and ergodic sources. 
The following definition formally states the underlying probability distribution and the corresponding random variables regarding the set of $K$ sources. 

\begin{definition}[Sources] \label{def:source}
Let $\mathbf{S}=[S_1,\cdots,S_K]^T\in[0:p-1]^K$ be a random vector associated with a joint probability mass function $p_{\mathbf{S}}(\cdot)$.
At each time $j\in[1:k]$, $\mathbf{s}[j]$ is assumed to be independently drawn from $p_{\mathbf{S}}(\cdot)$.
\end{definition}

As a special case in Definition \ref{def:source}, we will consider the following doubly symmetric binary sources throughout the paper.
\begin{definition}[Doubly Symmetric Binary Sources] \label{def:dsbs}
Assume $K=2$. Denote the doubly symmetric binary sources with the associated probability $\alpha$ by  DSBS($\alpha$). 
Let $\operatorname{Bern}(a)$ be the Bernoulli distribution with the probability $a$.
For DSBS($\alpha$), $S_1$ follows $\operatorname{Bern}(1/2)$ and $S_2=S_1\oplus Z$, where $Z$ follows $\operatorname{Bern}(\alpha)$ and is independent of $S_1$.
\end{definition}

Let $f(\mathbf{S})$ denote the desired function induced by the random source vector $\mathbf{S}$.
The following two definitions define random variables associated with the desired function, which will be used throughout the paper.

\begin{definition}[Arithmetic Sum Function Induced by $\mathbf{S}$]  \label{def:rv1}
Define $U_l=\sum_{i=1}^Ka_{li}S_i$ for $l\in[1:L]$, which are the random variables associated with the arithmetic sum function.
Then $f(\mathbf{S})=(U_1,\cdots,U_L)$ for the arithmetic sum function.
\end{definition}

\begin{definition}[Type Function Induced by $\mathbf{S}$]  \label{def:rv2}
Define $B_l=\sum_{i=1}^K\mathbf{1}_{S_i=l}$ for $l\in[0:p-1]$, which are the random variables associated with the type function.
Then $f(\mathbf{S})=(B_0,\cdots,B_{p-1})$ for the type function. 
\end{definition}

\begin{remark}[Worst Case Sources]\label{remark:worst_sources}
Note that $H(f(\mathbf{S}))$ is upper bounded by $\min\{K\log p, L\log (p^2 K)\}$ for the arithmetic sum function and $\min\{K\log p, p\log (K+1)\}$ for the type function.
For both cases, $H(f(\mathbf{S}))$ scales as the order of $\log K$ as the number of sources $K$ increases.
\end{remark}

Associated with $G=(V,E)$, we consider two classes of Gaussian channels, which are formally stated in the following two definitions.

\begin{definition}[Orthogonal Gaussian Networks] \label{def:channel1}
For this case, we assume Gaussian point-to-point channels  with no broadcast and no multiple-access for each $(u,v)\in E$.
That is, the length-$n$ time-extended input--output is given by
\begin{equation} \label{eq:in_out_ptp}
\mathbf{y}_{u,v}=h_{u,v}\mathbf{x}_{u,v}+\mathbf{z}_{u,v},
\end{equation} 
where the elements of $\mathbf{z}_{u,v}$ are independently drawn from $\mathcal{N}(0,1)$.
Each transmit signal should satisfy $\frac{1}{n}\|\mathbf{x}_{u,v}\|^2\leq P$ for all $(u,v)\in E$.
For notational simplicity, we will use the subscript $(\cdot)_{\sf{ ptp}}$ to denote orthogonal Gaussian networks.
\end{definition}

\begin{definition}[Gaussian Networks With Multiple-Access] \label{def:channel2}
For this case, we assume Gaussian multiple-access channels with no broadcast from $u\in \Gamma_{in}(v)$ to each $v\in V$.
That is, the length-$n$ time-extended input--output is given by
\begin{equation} \label{eq:in_out}
\mathbf{y}_{v}=\sum_{u\in \Gamma_{in}(v)}h_{u,v}\mathbf{x}_{u,v}+\mathbf{z}_v,
\end{equation}
where the elements of $\mathbf{z}_v$ are independently drawn from $\mathcal{N}(0,1)$.
Each transmit signal should satisfy $\frac{1}{n}\|\mathbf{x}_{u,v}\|^2\leq P$ for all $(u,v)\in E$.
For notational simplicity, we will use the subscript $(\cdot)_{\sf{mac}}$ to denote Gaussian networks with multiple-access.
\end{definition}

\begin{remark}[Bit-Pipe Wired Networks] \label{re:wired}
The considered orthogonal Gaussian network is almost equivalent to a bit-pipe wired network in the sense that it can be easily converted into a bit-pipe wired network by using capacity-achieving point-to-point channel codes. Nevertheless, we will state this paper based on Gaussian networks assuming orthogonal components defined in Definitions \ref{def:channel1} and \ref{def:channel2}. 
\end{remark}

\begin{remark}[Single-Hop Networks] For notational simplicity, we will use the following simplified notation for the single-hop case.
For  orthogonal Gaussian single-hop networks, we rewrite the length-$n$ time-extended input--output as
\begin{equation} \label{eq:gaussian_ptp}
\mathbf{y}_i= h_i \mathbf{x}_i+\mathbf{z}_i,
\end{equation}
where $i\in[1:K]$.
For  Gaussian single-hop networks with multiple-access or  Gaussian MACs, we rewrite the length-$n$ time-extended input--output as
\begin{equation} \label{eq:gaussian_mac}
\mathbf{y}= \sum_{i=1}^K h_i \mathbf{x}_i+\mathbf{z}.
\end{equation}

\end{remark}

\subsection{Computation Capacity}
Based on the above network model, the length-$n$ block code for orthogonal Gaussian networks is defined as follows, where $\mathbf{y}^{a}$ denotes $y[1],\cdots,y[a]$
\begin{itemize}
\item (Sender Encoding) The $i$th sender $t_i$ transmits $x^{(t)}_{t_i,w}=\psi^{(t)}_{t_i,w}\left(\mathbf{s}_i, \{\mathbf{y}_{u,t_i}^{t-1}\}_{u\in\Gamma_{in}(t_i)}\right)$ for $t\in[1:n]$ to node $w\in \Gamma_{out}(t_i)$.

\item (Relay Encoding) Node $v\notin \{t_i\}_{i\in[1:K]}\bigcup$\{d\} transmits $x^{(t)}_{v,w}=\psi^{(t)}_{v,w}\left( \{\mathbf{y}_{u,v}^{t-1}\}_{u\in\Gamma_{in}(v)}\right)$ for $t\in[1:n]$ to node $w\in \Gamma_{out}(v)$.

\item (Decoding) The receiver $d$ estimates $\hat{f}(\mathbf{s}[j])=\varphi^{(j)}\left(\{\mathbf{y}_{u,d}\}_{u\in\Gamma_{in}(d)}\right)$ for $j\in[1:k]$. 
\end{itemize}

Similarly, we can define the length-$n$ block code for Gaussian networks with multiple-access. Specifically,  $x^{(t)}_{t_i,w}=\phi^{(t)}_{t_i,w}\left(\mathbf{s}_i, \mathbf{y}_{t_i}^{t-1}\right)$, $x^{(t)}_{v,w}=\phi^{(t)}_{v,w}\left( \mathbf{y}_v^{t-1}\right)$ for $v\notin \{t_i\}_{i\in[1:K]}\bigcup\{d\}$, and $\hat{f}(\mathbf{s}[j])=\varphi^{(j)}\left(\mathbf{y}_{d}\right)$.

The probability of error is defined by $P^{(n)}_e=\Pr\left[\bigcup_{j=1}^k\hat{f}(\mathbf{s}[j])\neq f(\mathbf{s}[j])\right]$.
We then define the computation capacity as the follow.

\begin{definition}[Computation Capacity] \label{def:computation_rate}
The computation rate $R:=\frac{k}{n}$ is said to be achievable if there exists a sequence of length-$n$ block codes such that $P^{(n)}_{e}$ converges to zero as $n$ increases.
The computation capacity is the maximum over all achievable computation rates.
\end{definition}

From Definition \ref{def:computation_rate}, the computation rate is the number of reliably computable functions per channel use.

\section{Preliminaries}

Before stating our main results, we first introduce previous work that is closely related to our work in Sections \ref{subsec:dist_source_coding} to \ref{subsec:com_mac}. 
For comparison, we introduce a cut-set upper bound in Section \ref{subset:cut_set} and a separation-based lower bound in Section \ref{subsec:sep_com}.

\begin{figure}[t!]
\begin{center}
\includegraphics[scale=1.3]{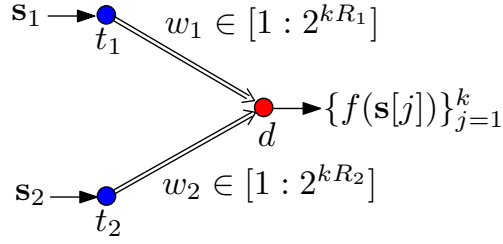}
\end{center}
\vspace{-0.15in}
\caption{Two-user distributed source coding for function computation.}
\label{figs:dist_source_coding}
\vspace{-0.1in}
\end{figure}

\subsection{Distributed Source Coding} \label{subsec:dist_source_coding}
Figure \ref{figs:dist_source_coding} illustrates the two-user distributed source coding for function computation.
Two senders respectively observe the length-$k$ sources $\mathbf{s}_1$ and $\mathbf{s}_2$ and deliver some information for function computing via messages $w_1\in[1:2^{kR_1}]$ and $w_2\in[1:2^{kR_2}]$. The receiver wishes to compute the desired function $\{f(\mathbf{s}[j]\}_{j=1}^k$ based on  $(w_1,w_2)$. 

The optimal distributed lossless source coding for $f(\mathbf{s}[j])=\mathbf{s}[j]$ in Fig. \ref{figs:dist_source_coding} has been solved by Slepian and Wolf \cite{Slepain:73}. For the two-user case, the Slepian--Wolf rate region is the set of all rate pairs $(R_1,R_2)$ satisfying
\begin{align} \label{eq:slepain_wolf}
R_1&\geq H(S_1|S_2),\nonumber\\
R_2&\geq H(S_2|S_1),\nonumber\\
R_1+R_2&\geq H(S_1,S_2).
\end{align}
Obviously, the above rate region is also an achievable rate region for any desired function.
It was proved by Csisz\'{a}r in \cite{Csiszar:82} that the same Slepian--Wolf rate region is achievable by linear source coding. 

The potential of the linear source coding has been first captured by K\"{o}rner and Marton in \cite{Korner:79} in the context of the modulo-two sum computation. Consider binary field sources $\mathbf{s}_1\in\mathbb{F}_2^k$ and $\mathbf{s}_2\in\mathbb{F}_2^k$ and $\{f(\mathbf{s}[j])=s_1[j]\oplus s_2[j]\}_{j=1}^k$. It was proved in \cite{Korner:79} that the set of all rate pairs $(R_1,R_2)$ satisfying 
\begin{align} \label{eq:linear_binning}
R_1&\ge H(S_1\oplus S_2),\nonumber\\
R_2&\ge H(S_1\oplus S_2)
\end{align}
is achievable by linear source coding.
A simple outer bound shows a necessary condition on an achievable $(R_1,R_2)$ as
\begin{align} \label{eq:outer}
R_1&\ge H(S_1|S_2),\nonumber\\
R_2&\ge H(S_1|S_2),\nonumber\\
R_1+R_2&\ge H(S_1\oplus S_2).
\end{align}

\begin{example}[Modulo-Two Sum of DSBS($\alpha$)]
For DSBS($\alpha$), the K\"{o}rner--Marton rate region \eqref{eq:linear_binning} shows that any rate pair satisfying $R_1\geq H_2(\alpha)$ and $R_2\geq H_2(\alpha)$ is achievable, which coincides with the outer bound in \eqref{eq:outer}.
\end{example}

Unfortunately, the optimal rate region for the modulo-two sum computation of arbitrarily correlated binary sources is unknown. 
The outer bound in \eqref{eq:outer} does not coincide with the convex hull of the union of the Slepian--Wolf rate region \eqref{eq:slepain_wolf} and the K\"{o}rner--Marton rate region \eqref{eq:linear_binning}.
A more general achievability containing both the Slepian--Wolf rate region \eqref{eq:slepain_wolf} and  the K\"{o}rner--Marton rate region \eqref{eq:linear_binning} has been proposed by Ahlswede and Han in \cite[Section VI]{Ahlswede:83}.
Let $W_1$ and $W_2$ be auxiliary random variables that form a Markov chain $W_1- S_1- S_2-W_2$. 
Ahlswede and Han showed that the set of all rate pairs $(R_1,R_2)$ satisfying
\begin{align}
R_1&\geq I(W_1;S_1|W_2)+H(S_1\oplus S_2|W_1,W_2),\nonumber\\
R_2&\geq I(W_2;S_2|W_1)+H(S_1\oplus S_2|W_1,W_2),\nonumber\\
R_1+R_2&\geq I(W_1,W_2;S_1,S_2)+2H(S_1\oplus S_2|W_1,W_2)
\end{align}
is achievable. An example that this new rate region strictly enlarges the convex hull of the union of the Slepian--Wolf rate region \eqref{eq:slepain_wolf} and the K\"{o}rner--Marton rate region \eqref{eq:linear_binning} was also provided in \cite[Example 4]{Ahlswede:83}.

\begin{figure}[t!]
\begin{center}
\includegraphics[scale=1.3]{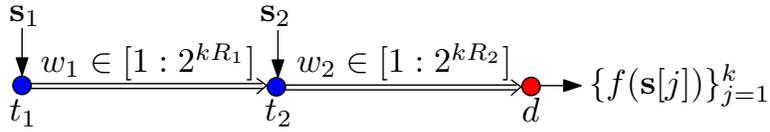}
\end{center}
\vspace{-0.15in}
\caption{Two-user cascade source coding for function computation.}
\label{figs:cascade_source_coding}
\vspace{-0.1in}
\end{figure}

\subsection{Cascade Source Coding}
Figure \ref{figs:cascade_source_coding} illustrates the two-user cascade source coding for function computation.
The first sender observes the length-$k$ source $\mathbf{s}_1$ and delivers some information for function computing via message $w_1\in[1:2^{kR_1}]$. The second sender observes $(\mathbf{s}_2,w_1)$ and again delivers some information of $(\mathbf{s}_2,w_1)$ for function computing via message $w_2\in[1:2^{kR_2}]$.
The receiver wishes to compute the desired function $\{f(\mathbf{s}[j])\}_{j=1}^k$ based on   $w_2$. 

The cascade source coding for function computation has been studied in \cite{Cuff:09} in the context of lossy computation.
For the lossless case depicted in Fig. \ref{figs:cascade_source_coding}, the computation capacity for a general function has been shown to be represented by the conditional graph entropy \cite{Orlitsky:01}.

Again, consider arbitrarily correlated binary field sources $\mathbf{s}_1\in\mathbb{F}_2^k$ and $\mathbf{s}_2\in\mathbb{F}_2^k$ and the modulo-two sum computation of these two sources. 
In this case, the optimal rate region is given by the set of all $(R_1,R_2)$ satisfying 
\begin{align}
R_1&\ge H(S_1|S_2),\nonumber\\
R_2&\ge H(S_1\oplus S_2),
\end{align}
which can be attained from a simple application of Slepian--Wolf  source coding.

\begin{figure}[t!]
\begin{center}
\includegraphics[scale=1.3]{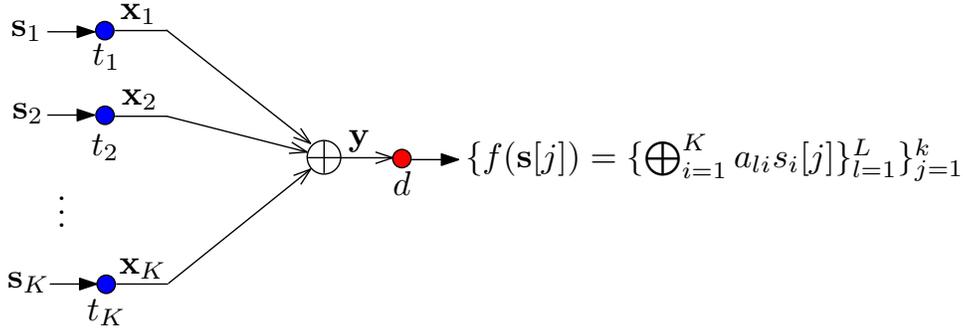}
\end{center}
\vspace{-0.15in}
\caption{Modulo-$p$ sum computation over the modulo-$p$ sum channel.}
\label{figs:mod_p_sum}
\vspace{-0.1in}
\end{figure}

\subsection{Modulo-$p$ Sum Computation Over MAC} \label{subsec:com_mac}
The potential of linear source coding has been also captured by Nazer and Gastpar in the context of the computation over MAC \cite{Nazer:07}.
Figure \ref{figs:mod_p_sum} illustrates the modulo-$p$ sum computation over the deterministic modulo-$p$ sum channel and a more general case can be found in \cite[Theorem 1]{Nazer:07}.
The $i$th sender observes $\mathbf{s}_i\in\mathbb{F}_p^k$ and the receiver wishes to compute the modulo-$p$ sum function, i.e., $f(\mathbf{s}[j])=\{\bigoplus_{i=1}^K a_{li}s_i[j]\}_{l=1}^L$, $a_{li}\in\mathbb{F}_p$. The length-$n$ time-extended input--output of the modulo-$p$ sum channel is given by $\mathbf{y}=\bigoplus_{i=1}^K\mathbf{x}_i$, where $\mathbf{x}_i\in \mathbb{F}^n_p$  and $p$ is assumed to be a prime number.

To compress multiple modulo-$p$ sum functions, which are in general correlated to each other, Nazer and Gastpar applied linear Slepian--Wolf source coding as introduced in the following lemma.
  
\begin{lemma}[Csisz\'{a}r \cite{Csiszar:82}] \label{thm:csiszar}
Let $(\mathbf{v}_1,\cdots,\mathbf{v}_L)$ be the set of length-$k$ sources, independently drawn from some joint probability mass function $p_{V_1,\cdots,V_L}(\cdot)$. 
For any point in the Slepian--Wolf rate region, there exist matrices $\mathbf{H}_1,\cdots,\mathbf{H}_L$ of size $n_l\times k$, respectively, taking values over a finite field with associated decoding function that can be used to compress the sources in a distributed fashion with $\Pr[(\hat{\mathbf{v}}_1,\cdots,\hat{\mathbf{v}}_L)\neq (\mathbf{v}_1,\cdots,\mathbf{v}_L)]\to 0$ as $k$ increases.
\end{lemma}

\begin{theorem}[Nazer--Gastpar \cite{Nazer:07}] \label{thm_comp_mac}
Consider the modulo-$p$ sum computation over the modulo-$p$ sum channel depicted in Fig. \ref{figs:mod_p_sum}. 
Let $V_l=\bigoplus_{i=1}^K a_{li}S_i$.
Then the computation capacity is given by
\begin{equation}
R=\frac{\log p}{H(V_1,\cdots,V_L)}.
\end{equation}
\end{theorem}

For a better understanding, we briefly explain the achievability here.
From lemma \ref{thm:csiszar}, we set $\mathbf{H}_1,\cdots,\mathbf{H}_L$ of size $n_l\times k$, respectively, which corresponds to some point in the Slepian--Wolf rate region with sum rate $H(V_1,\cdots,V_L)$.
 
The $i$th sender transmits 
\begin{equation} \label{eq:trasmit_signal}
\mathbf{x}_i=[a_{1i}\mathbf{H}_1\mathbf{s}_i,\cdots,a_{Li}\mathbf{H}_L\mathbf{s}_i]^T
\end{equation}
for $i\in[1:K]$, where we set $n=\sum_{l=1}^Ln_l$.
Then \eqref{eq:trasmit_signal} yields
\begin{equation}
\mathbf{y}=\bigoplus_{i=1}^K\mathbf{x}_i=\left[\mathbf{H}_1\left(\bigoplus_{i=1}^Ka_{1i}\mathbf{s}_i\right),\cdots,\mathbf{H}_L\left(\bigoplus_{i=1}^Ka_{Li}\mathbf{s}_i\right)\right]^T.
\end{equation}
Hence, from Lemma \ref{thm:csiszar}, the receiver can recover $\{\bigoplus_{i=1}^Ka_{li}\mathbf{s}_i\}_{l\in[1:L]}$ with an arbitrarily small probability of error as $n$ increases if
\begin{equation} \label{eq:condition_k}
\left(\sum_{l=1}^L n_l\right)\log p=n\log p\geq kH(V_1,\cdots,V_L).
\end{equation}
Therefore, setting $k=\frac{n \log p}{H(V_1,\cdots,V_L)}$, which satisfies \eqref{eq:condition_k}, provides that
\begin{equation}
R=\frac{\log p}{H(V_1,\cdots,V_L)}
\end{equation}
is achievable.

\begin{figure}[t!]
\begin{center}
\includegraphics[scale=1.3]{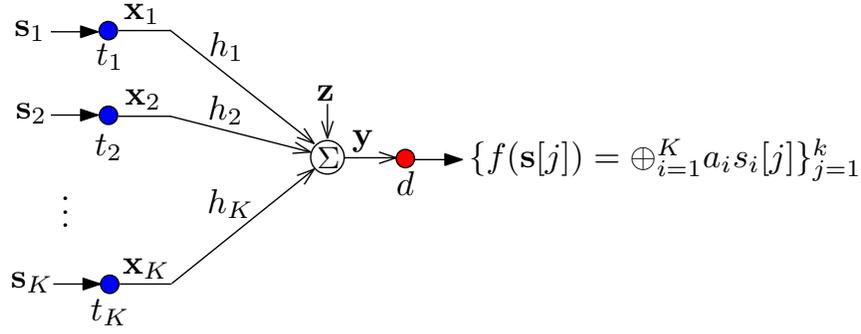}
\end{center}
\vspace{-0.15in}
\caption{Modulo-$p$ sum computation over  the Gaussian MAC.}
\label{figs:compute_forward}
\vspace{-0.1in}
\end{figure}

Nazer and Gastpar recently proposed compute-and-forward \cite{Nazer:11}, which provides a general framework for computing modulo-$p$ sum functions at multiple receivers over Gaussian channels.
We briefly describe compute-and-forward here with respect to a single receiver depicted in Fig. \ref{figs:compute_forward}. 
The $i$th sender observes $\mathbf{s}_i\in\mathbb{F}_p^k$ and the receiver wishes to compute the modulo-$p$ sum function $f(\mathbf{s}[j])=\bigoplus_{i=1}^K a_{i}s_i[j]$, $a_{i}\in\mathbb{F}_p$. 
Here $p$ is assumed to be a prime number.
The length-$n$ time-extended input--output is given by $\mathbf{y}=\sum_{i=1}^Kh_i\mathbf{x}_i+\mathbf{z}$, where the elements of $\mathbf{z}$ are independently drawn from $\mathcal{N}(0,1)$ and $\frac{1}{n}\|\mathbf{x}_i\|^2\leq P$ for all $i\in[1:K]$.

\begin{theorem}[Nazer--Gastpar\cite{Nazer:11}] \label{thm:compute-forward}
Consider the modulo-$p$ sum computation over the Gaussian MAC depicted in Fig. \ref{figs:compute_forward}.
Let $\mathbf{a}=[a_1,\cdots,a_K]^T$ and $\mathbf{h}=[h_1,\cdots,h_K]^T$.
Then the receiver can decode $\{f(\mathbf{s}[j])\}_{j=1}^k=\bigoplus_{i=1}^K a_{i}\mathbf{s}_i$ reliably for $n$ sufficiently large if 
\begin{equation}
R=\frac{k}{n}\leq {\sf C}^+\left(\left(\|\mathbf{a}\|^2-\frac{P(\mathbf{h}^T\mathbf{a})^2}{1+P\|\mathbf{h}\|^2}\right)^{-1}\right)(\log p)^{-1}
\end{equation}
and $p$ is an increasing function of $n$ such that $\frac{n}{p}\to 0$ as $n\to \infty$.  
\end{theorem}

\begin{example}[Gaussian MAC When $h_i=a_i=1$] \label{ex:compute_forward_equal_ch}
Suppose that $h_i=a_i=1$ for all $i\in[1:K]$.
For this case, Theorem \ref{thm:compute-forward} yields that the receiver can decode $\bigoplus_{i=1}^K \mathbf{s}_i$ reliably for $n$ sufficiently large if $R\leq {\sf C}^{+}\left(\frac{1}{K}+P\right)(\log p)^{-1}$ and $p$ is an increasing function of $n$ satisfying that $p\to\infty$ as $n\to \infty$.
\end{example}

\subsection{Cut-Set Upper Bound} \label{subset:cut_set}
To describe a cut-set upper bound on the computation capacity, we first introduce the notation.
For a subset $\Sigma\subseteq[1:K]$, define $G(\Sigma)=(V(\Sigma),E(\Sigma))$ as the subgraph of $G$ consisting of the nodes having a direct path from at least one of the senders in $\{t_i\}_{i\in\Sigma}$. 
Let $\Lambda(\Sigma)$ denote the set of all cuts dividing all of the senders in $\{t_i\}_{i\in\Sigma}$ from the receiver $d$ on $G(\Sigma)$. 
Then the minimum-cut value for general discrete memoryless networks over the cuts in $\Lambda(\Sigma)$ on $G(\Sigma)$ is given by 
\begin{equation} \label{eq:min_cut}
\max_{p(\{x_v\}_{v\in V(\Sigma)})}\min_{\Omega\in \Lambda(\Sigma)} I(X_{\Omega};Y_{\Omega^c}|X_{\Omega^c}).
\end{equation}
If there exists an input cost constraint, $p(\{x_v\}_{v\in V(\Sigma)})$ should be set to satisfy the corresponding input cost constraint.

For orthogonal Gaussian networks, the minimum-cut value is given by
\begin{equation} \label{eq:min_cut_ptp}
\bar{C}_{\sf ptp}(\Sigma):=\min_{\Omega\in \Lambda(\Sigma) }\sum_{(u,v)\in E(\Sigma), u\in \Omega, v\in\Omega^c}{\sf C}\left(h^2_{u,v} P\right).
\end{equation}
Similarly, for Gaussian networks with multiple-access, the minimum-cut value is upper bounded by
\begin{equation} \label{eq:min_cut_mac}
\bar{C}_{\sf mac}(\Sigma):=\min_{\Omega\in \Lambda(\Sigma) }\sum_{v\in \Omega^c}{\sf C}\left(\Big(\sum_{u\in \Gamma_{in}(v), u\in \Omega}|h_{u,v}|\Big)^2 P\right).
\end{equation}

Since the desired function is locally computable \cite{Giridhar:05},  for any $\Sigma\subseteq[1:K]$, $f(\mathbf{s}[j])$ can be represented as $f'(\{s_i[j]\}_{i\in\Sigma})+f''(\{s_i[j]\}_{i\in [1:K]\setminus \Sigma})$.
Hence, by assuming that a genie provides $\{\mathbf{s}_i\}_{i\in [1:K]\setminus \Sigma}$ to the receiver, computing $\{f'(\{s_i[j]\}_{i\in\Sigma})\}_{j=1}^k$ at the receiver is enough to recover the desired function $\{f(\mathbf{s}[j])\}_{j=1}^k$. Therefore assuming full cooperation between the nodes in $\Omega$ and between the nodes in $\Omega^c$, from the source--channel separation theorem \cite{Shannon:48}, 
\begin{equation}  \label{eq:cutset_ptp}
R_{\sf ptp}\leq \min_{\Sigma\subseteq [1:K]}\frac{\bar{C}_{\sf ptp}(\Sigma)}{H(f(\mathbf{S})|\{S_i\}_{i\in[1:K]\setminus \Sigma})}
\end{equation}
for orthogonal Gaussian networks.
In the same manner, we have 
\begin{equation} \label{eq:cutset_int}
R_{\sf mac}\leq \min_{\Sigma\subseteq [1:K]}\frac{\bar{C}_{\sf mac}(\Sigma)}{H(f(\mathbf{S})|\{S_i\}_{i\in[1:K]\setminus \Sigma})}
\end{equation}
for Gaussian networks with multiple-access.

\subsection{Separation-Based Computation} \label{subsec:sep_com}
We generalize the notion of the separation-based computation introduced in \cite{Nazer:07}.
We refer to \cite[Section III]{Nazer:07} for the formal definition of the separation-based computation.
Roughly speaking, the separation-based computation means that a communication network is first transformed into an end-to-end bit-pipe channel by channel coding and then separately applied source coding for computing the desired function over the transformed end-to-end bit-pipe channel.  

Let $\mathbf{R}_{\sf{f}}$ denote the distributed compression rate region for computing $\{f(\mathbf{s}[j])\}_{j=1}^k$ (see \cite[Definition 8]{Nazer:07}) and $\mathbf{C}_{\sf{ptp}}$ denote the capacity region for orthogonal Gaussian networks, which can be represented as the set of all rate tuples such that
\begin{equation}
\sum_{i\in \Sigma} R_i\leq\bar{C}_{\sf ptp}(\Sigma) \mbox{ for all }\Sigma\subseteq[1:K].
\end{equation}
Then a computation rate $R^{(\sf sep)}_{\sf ptp}$ is achievable by separation if 
\begin{equation} \label{eq:sep_general}
\mathbf{R}_{\sf{f}}\cap\mathbf{C}'_{\sf{ptp}}\neq\emptyset,
\end{equation}
where $\mathbf{C}'_{\sf{ptp}}=\left\{\left(\frac{R_1}{R^{(\sf sep)}_{\sf ptp}},\cdots,\frac{R_K}{R^{(\sf sep)}_{\sf ptp}}\right):(R_1,\cdots,R_K)\in\mathbf{C}_{\sf{ptp}}\right\}$.

Similarly, we define the achievable computation rate $R_{\sf mac}$ for Gaussian networks with multiple-access by separation. Specifically, $R^{(\sf sep)}_{\sf mac}$ is achievable by separation if 
\begin{equation}
\mathbf{R}_{\sf{f}}\cap\mathbf{C}'_{\sf{mac}}\neq\emptyset,
\end{equation}
where $\mathbf{C}'_{\sf{mac}}=\left\{\left(\frac{R_1}{R^{(\sf sep)}_{\sf mac}},\cdots,\frac{R_K}{R^{(\sf sep)}_{\sf mac}}\right):(R_1,\cdots,R_K)\in\mathbf{C}_{\sf{mac}}\right\}$ and $\mathbf{C}_{\sf{mac}}$ denotes the capacity region for Gaussian networks with multiple-access, which is upper bounded by 
\begin{equation}
\sum_{i\in \Sigma} R_i\leq\bar{C}_{\sf mac}(\Sigma) \mbox{ for all }\Sigma\subseteq[1:K].
\end{equation}

\begin{example}[I.I.D. Sources and Symmetric MAC]
Suppose that the sources are i.i.d., i.e., $p_{\mathbf{S}}(\cdot)=\prod_{i=1}^Kp_{S_i}(\cdot)$ and $p_{S_i}(\cdot)=p_{S}(\cdot)$. 
Then from \cite[Lemma 1]{Nazer:07} (also see \cite[Example 1]{Nazer:07}), $\mathbf{R}_{\sf f}$ is given by all rate tuples such that $R_i\geq H(S)$ for both the  arithmetic sum function and the type function computation.
Therefore, for symmetric MAC, see the definition in \cite[Definition 11]{Nazer:07}, an achievable computation rate by separation is upper bounded by
\begin{equation} \label{eq:upper_separation_ptp}
R^{(\sf sep)}_{\sf ptp}\leq \frac{\bar{C}_{\sf ptp}([1:K])}{KH(S)}
\end{equation}
for orthogonal Gaussian single-hop networks and 
\begin{equation} \label{eq:upper_separation_mac}
R^{(\sf sep)}_{\sf mac}\leq \frac{\bar{C}_{\sf mac}([1:K])}{KH(S)}
\end{equation}
for Gaussian MACs.
\end{example}

\section{Main Results}
In this section, we state our main results.
For a better understanding, we first provide a high level description of the proposed approach based on the single arithmetic sum computation over the Gaussian MAC in Section \ref{subsec:outline}. We then state our results for orthogonal Gaussian networks and Gaussian networks with multiple-access in Sections \ref{subsec:result1} and \ref{subsec:result2}, respectively.

\subsection{Main Idea} \label{subsec:outline}
We begin this section by explaining the essence of how to compute a single arithmetic sum, i.e., $\{f(\mathbf{s}[j])=\sum_{i=1}^K s_i[j]\}_{j=1}^k$, over the Gaussian MAC with equal channel gains.
For notational simplicity, we rewrite the length-$n$ time-extended input--output as $\mathbf{y}=\sum_{i=1}^K \mathbf{x}_i+\mathbf{z}$.
We first apply compute-and-forward in \cite{Nazer:11} to transform the length-$n$ Gaussian MAC into the following length-$m$ modulo-$q$ sum channel:
\begin{align} \label{eq:modulo_sum_ch11}
\mathbf{y}'=\bigoplus_{i=1}^K \mathbf{x}'_i, 
\end{align}
where $\mathbf{x}'_i\in{\mathbb{F}_q}^m$. Here $q$ is  set to be the largest prime number among $[1:n\log n]$ and 
\begin{equation} \label{eq:m_value11}
m=n{\sf C}^{+}\left(\frac{1}{K}+P\right)(\log q)^{-1}.
\end{equation}
Specifically, from Theorem \ref{thm:compute-forward} (also see Example \ref{ex:compute_forward_equal_ch}), by treating $\mathbf{y}'$ in \eqref{eq:modulo_sum_ch11} as the desired function, we can construct the above modulo-$q$ sum channel.

Now consider the computation over the transformed modulo-$q$ sum channel.
The key observation is that utilizing a small portion of input finite field elements and then computing the corresponding modulo-$q$ sum can attain the desired arithmetic sum. Furthermore, linear source coding in Lemma \ref{thm:csiszar} (for this case, $L=1$) can compensate the inefficiency of utilizing only a small portion of input finite field elements by compressing the corresponding modulo-$q$ sum in a distributed manner.

Let $g(\cdot)$ denote the mapping from a subset of integers $[0:q-1]$ to the corresponding finite field $\mathbb{F}_q$.
Define $s'_i[j]=g(s_i[j])$ and $U'=\bigoplus_{i=1}^Kg(S_i)$.
Suppose that the $i$th sender observes $\mathbf{s}'_i=[s'_i[1],\cdots s'_i[k]]^T$, which can be obtained from $\mathbf{s}_i$, and the receiver wishes to compute $\mathbf{u}'=\bigoplus_{i=1}^K \mathbf{s}_i'$.
Then, from Theorem \ref{thm_comp_mac}, the receiver can compute $\mathbf{u}'$ reliably for $m$ sufficiently large (equivalently, for $n$ sufficiently large) if 
\begin{equation}
\frac{k}{m}\leq \frac{\log q}{H(U')}.
\end{equation}
Hence by setting $k=\frac{m\log q}{H(U')}$, the computation rate 
\begin{align} \label{eq:compuration_rate_ex}
R&=\frac{k}{n}\nonumber\\
&=\frac{m\log q}{nH(U')}\nonumber\\
&=\frac{{\sf C}^{+}\left(\frac{1}{K}+P\right)}{H(U')}
\end{align}
is achievable for the desired function $\mathbf{u}'$, where the last equality follows from \eqref{eq:m_value11}.
Since there exists $n_0\geq 0$ such that  $q>(p-1)^2K$ for all $n\geq n_0$ ($q$ is  the largest prime number among $[1:n\log n]$), we have
\begin{align}
\mathbf{u}'&=\left[g\left(\sum_{i=1}^Ks_i[1]\right),\cdots,g\left(\sum_{i=1}^Ks_i[k]\right)\right]^T,\nonumber\\
U'&=g\left(\sum_{i=1}^KS_i\right)
\end{align}
for $n$ sufficiently large.
Since $g(\cdot)$ has one-to-one correspondence, the receiver can compute the arithmetic sum $\sum_{i=1}^K\mathbf{s}_i$ from $\mathbf{u}'$.
Finally, from the fact that $H(g(\sum_{i=1}^KS_i))=H(\sum_{i=1}^KS_i)$, the achievable computation rate for the desired function $\sum_{i=1}^K\mathbf{s}_i$  is given by
\begin{equation} \label{eq:com_mac_ex}
R=\frac{{\sf C}^{+}\left(\frac{1}{K}+P\right)}{H(\sum_{i=1}^KS_i)}.
\end{equation}


\begin{example}[Arithmetic Sum of I.I.D. Binary Sources] \label{ex:sum_binary}
Suppose that $S_i$'s are independently and uniformly drawn from $\{0,1\}$ and the receiver wishes to compute $\{f(\mathbf{s}[j])=\sum_{i=1}^K s_i[j]\}_{j=1}^k$.
Let $U=\sum_{i=1}^KS_i$. 
Then, from \eqref{eq:com_mac_ex}, $R=\frac{{\sf C}^{+}\left(\frac{1}{K}+P\right)}{H(U)}$ is achievable, where $p_{U}(x)={K \choose x} 2^{-K}$.
On the other hand, an achievable computation rate by separation is upper bounded by $\frac{{\sf C}(K^2 P)}{K}$ from \eqref{eq:upper_separation_mac}. Lastly, the cut-set upper bound in \eqref{eq:cutset_int} shows that an achievable computation rate is upper bounded by $\frac{{\sf C}(K^2 P)}{H(U)}$. Figure \ref{figs:iid_binary_exam} plots these three rates with respect to $K$.
Since $H(U)$ scales as the order of $\log K$ as $K$ increases \cite[Lemma 2.1]{Chang:79}, our computation-based rate decreases as the order of $\log K$, while the separation-based rate decreases almost linearly with an increasing $K$. Therefore, the rate gap between the computation-based and separation-based  approaches becomes significant as $K$ increases.
\end{example}

\begin{figure}[t!]
\begin{center}
\includegraphics[scale=1]{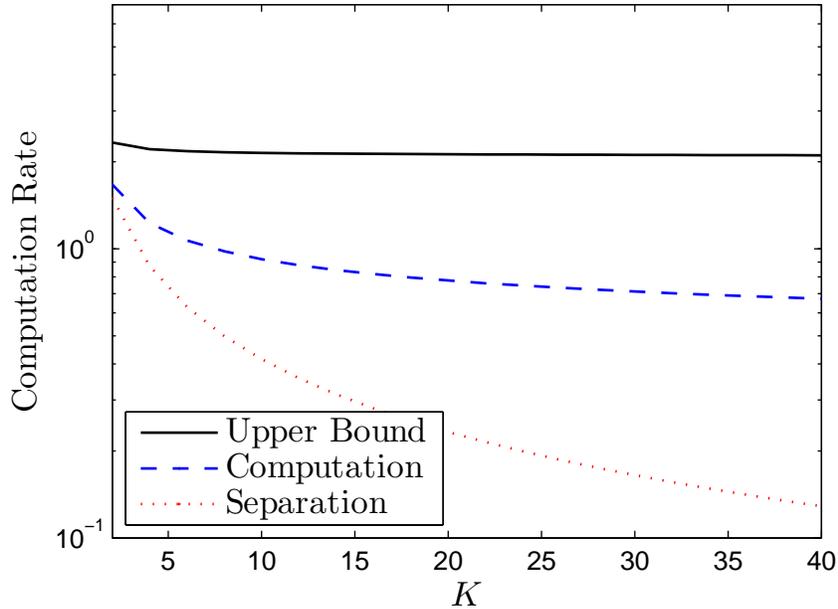}
\end{center}
\vspace{-0.15in}
\caption{Computation of $\sum_{i=1}^K S_i$ for the $K$-user Gaussian MAC with equal channel gains  when $P=15$ dB.}
\label{figs:iid_binary_exam}
\vspace{-0.1in}
\end{figure}

\subsection{Orthogonal Gaussian Networks} \label{subsec:result1}

We first state our main result for orthogonal Gaussian single-hop networks and then extend it to general orthogonal Gaussian networks. We also demonstrate a class of networks that achieves the computation capacity. 


\subsubsection{Single-hop networks}
Consider orthogonal Gaussian single-hop networks in which the length-$n$ time-extended input--output is given by \eqref{eq:gaussian_ptp}.
For orthogonal Gaussian networks, we abstract each Gaussian channel into the corresponding error-free bit-pipe channel using point-to-point capacity-achieving codes.
Then the problem is equivalent to the distributed source coding problem.
In \cite[Section VI]{Ahlswede:83}, a general achievability containing the Slepian--Wolf rate region and the K\"{o}rner--Marton rate region has been proposed for the modulo-two sum computation of binary sources, which has been introduced in Section \ref{subsec:dist_source_coding}.
The coding scheme in \cite[Section VI]{Ahlswede:83} can be straightforwardly generalized to more than two users and general finite field sources, and more importantly, to the arithmetic sum and type computation. 

\begin{theorem}[Orthogonal Gaussian Single-Hop Networks] \label{thm:com_rate_orthogonal_unequal}
Consider the orthogonal Gaussian single-hop network.
Let $W_i$, $i\in[1:K]$, be an auxiliary random variable that forms a Markov chain $W_i-S_i- \{S_j,W_j\}_{j\in[1:K]\setminus \{i\}}$. 
Then any computation rate satisfying 
\begin{align} \label{eq:com_rate_orthogonal_unequal}
R\leq\frac{\sum_{i\in \Sigma}{\sf C}(h_i^2 P)}{I(\{W_i\}_{i\in\Sigma};\mathbf{S}|\{W_i\}_{i\in[1:K]\setminus\Sigma})+|\Sigma|H(f(\mathbf{S})|W_1,\cdots, W_K)}
\end{align}
for all $\Sigma\subseteq [1:K]$ is achievable, where $f(\mathbf{S})=(U_1,\cdots,U_L)$ for the arithmetic sum function and $f(\mathbf{S})=(B_0,\cdots,B_{p-1})$ for the type function.
\end{theorem}
\begin{proof}
We refer to Appendix I for the proof.
\end{proof}

Notice that if we set $W_i=\emptyset$ for all $i\in[1:K]$, then Theorem \ref{thm:com_rate_orthogonal_unequal} provides  
\begin{align}
R\leq\frac{\sum_{i\in \Sigma}{\sf C}(h_i^2 P)}{|\Sigma|H(f(\mathbf{S}))}
\end{align}
for all $\Sigma\subseteq [1:K]$, which corresponds to the achievable computation rate from linear source coding for computing.
On the other hand, if we set $W_i=S_i$ for all $i\in[1:K]$, then Theorem \ref{thm:com_rate_orthogonal_unequal} provides  
\begin{align} \label{eq:slepian_wolf_comp}
R\leq\frac{\sum_{i\in \Sigma}{\sf C}(h_i^2 P)}{H(\{S_i\}_{i\in\Sigma}|\{S_i\}_{i\in[1:K]\setminus\Sigma})}
\end{align}
for all $\Sigma\subseteq [1:K]$, which corresponds to the achievable computation rate by separation.

Unlike the computation over MAC, for example, see Section \ref{subsec:outline}, the following example shows that without multiple-access component linear source coding cannot improve the computation rate achievable by separation if the sources are independent. 

\begin{example}[Independent Sources]
Suppose that the sources are independent to each other, i.e., $p_{\mathbf{S}}(\cdot)=\prod_{i=1}^Kp_{S_i}(\cdot)$. 
Then \eqref{eq:slepian_wolf_comp} yields $R\leq\frac{\sum_{i\in \Sigma}{\sf C}(h_i^2 P)}{\sum_{i\in\Sigma}H(S_i)}$.
For the arithmetic sum and type functions, $H(f(\mathbf{S}))\geq H(S_i)$ for all $i\in[1:K]$ when the sources are independent.
Hence, the separation-based computation always outperforms the computation based on linear source coding for independent sources.
\end{example}

When the sources are correlated, however, linear source coding is  useful even without multiple-access component. For DSBS($\alpha$), for example, the computation based on linear source coding  outperforms the separation-based computation when $\alpha$ is small, which can be verified in the following example.

\begin{example}[Arithmetic Sum of DSBS($\alpha$)]
Suppose that $K=2$, $h_1=h_2=1$, and the sources follows DSBS($\alpha$).
The receiver wishes to compute $\{f(\mathbf{s}[j])=s_1[j]+s_2[j]\}_{j=1}^k$.
From Theorem \ref{thm:com_rate_orthogonal_unequal}, let $W_1=S_1\oplus Z_1$ and $W_2=S_2\oplus Z_2$, where $Z_1$ and $Z_2$ are independent and follow $\operatorname{Bern}(\beta)$.
Then, due to symmetry, any computation rate satisfying
\begin{align} \label{eq:hybrid}
R&\leq\frac{{\sf C}(P)}{I(W_1;S_1,S_2|W_2)+H(S_1+S_2|W_1, W_2)},\nonumber\\
R&\leq\frac{2{\sf C}(P)}{I(W_1,W_2;S_1,S_2)+2H(S_1+S_2|W_1, W_2)}
\end{align}
is achievable. Figure \ref{figs:dsbs_orthogonal} plots the computation rate in \eqref{eq:hybrid}.
For comparison, we also plot the computation rates achievable by computation (setting $W_1=S_1$ and $W_2=S_2$ in \eqref{eq:hybrid}) and separation (setting $W_1=W_2=\emptyset$), respectively.
The cut-set upper bound is given by $R\leq\min\{{\sf C}(P)/H(Z),2{\sf C}(P)/(H(S_1+S_2))\}$ \cite{Korner:79}.
As shown in the figure, the computation based on linear source coding is indeed helpful even if there is no multiple-access component.
For this symmetric source, the hybrid approach in Theorem \ref{thm:com_rate_orthogonal_unequal} provides the maximum of the two computation rates achievable by the computation and separation schemes.
\end{example}

\begin{figure}[t!]
\begin{center}
\includegraphics[scale=1]{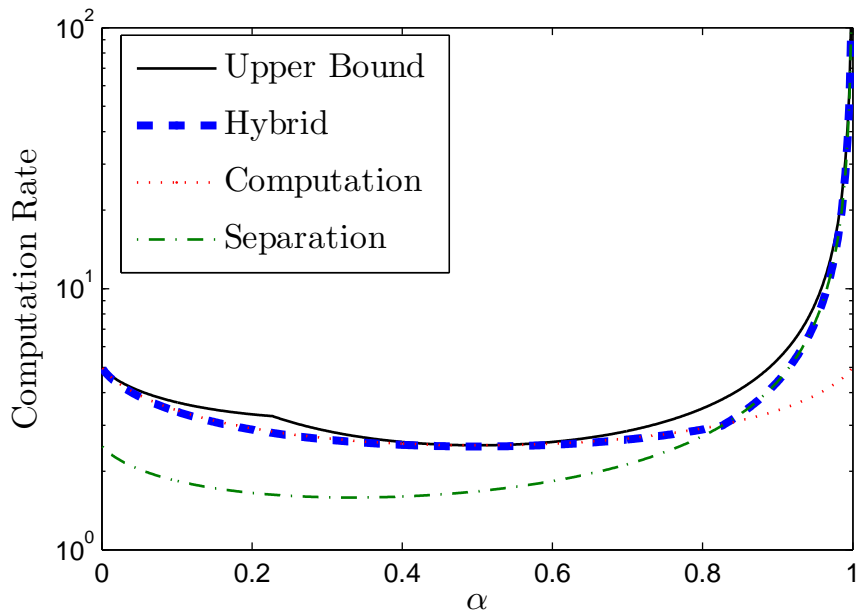}
\end{center}
\vspace{-0.15in}
\caption{Computation of $S_1+S_2$ for the two-user orthogonal Gaussian single-hop network with equal channel gains when $P=15$ dB.}
\label{figs:dsbs_orthogonal}
\vspace{-0.1in}
\end{figure}

\subsubsection{General networks} \label{sec:com_networks}
Now consider general orthogonal Gaussian networks. 
For function computation, the previous work in \cite{Nazer:07,Nazer:11} has exploited the similarity between the channel and the desired function.
Specifically, modulo-$p$ sum computation over (noisy) modulo-$p$ sum channel or Gaussian MAC has been considered.
It seems that channel's multiple-access or superposition property is essentially required to compute sum-type functions more efficiently than the separation-based computation.
We demonstrate that, even for orthogonal channels with no multiple-access component, relaying based on linear network coding provides an efficient end-to-end interface for function computation, which yields the following theorem.

\begin{theorem}[Orthogonal Gaussian Networks] \label{thm:orthogonal_net}
Consider the orthogonal Gaussian network. 
Then the computation rate 
\begin{equation} \label{eq:com_rate_ptp}
R_{\sf ptp}=\frac{\min_{i\in [1:K]}\bar{C}_{\sf ptp}(\{i\})}{H(f(\mathbf{S}))}
\end{equation} 
is achievable, where $f(\mathbf{S})=(U_1,\cdots,U_L)$ for the arithmetic sum function and $f(\mathbf{S})=(B_0,\cdots,B_{p-1})$ for the type function.
\end{theorem}
\begin{proof}
We refer to Section \ref{subsec:ortho_com_networks} for the proof.
\end{proof}

For convenience, denote the cut-set upper bound in \eqref{eq:cutset_ptp} for $\Sigma=[1:K]$ as 
\begin{equation} \label{eq:r_u_ptp}
R^{(u)}_{\sf ptp}:=\frac{\bar{C}_{\sf ptp}([1:K])}{H(f(\mathbf{S}))}
\end{equation}
and the achievable computation rate in Theorem \ref{thm:orthogonal_net} as 
\begin{equation} \label{eq:r_l_ptp}
R^{(l)}_{\sf ptp}:=\frac{\min_{i\in [1:K]}\bar{C}_{\sf ptp}(\{i\})}{H(f(\mathbf{S}))}.
\end{equation}

\begin{figure}[t!]
\begin{center}
\includegraphics[scale=1.3]{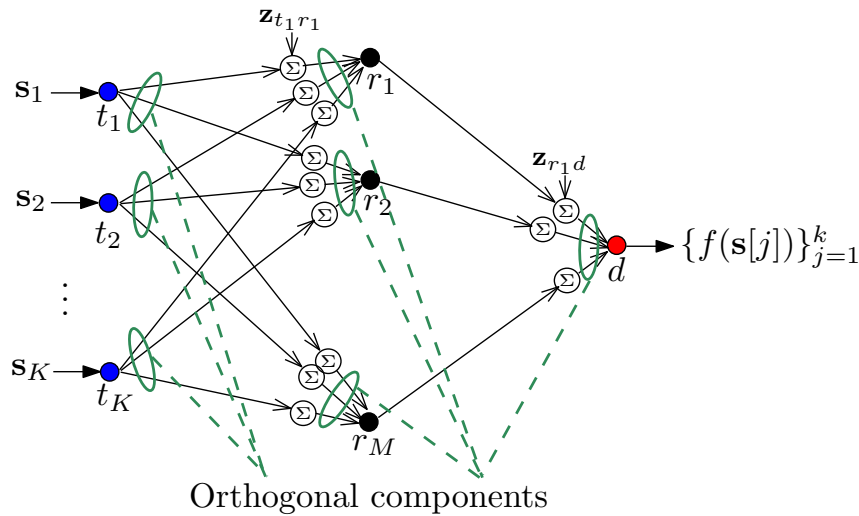}
\end{center}
\vspace{-0.15in}
\caption{An example of orthogonal Gaussian networks that achieves the computation capacity, where the channel coefficients are equal to one on all links.}
\label{figs:example_network1}
\vspace{-0.1in}
\end{figure}

\begin{remark}[Computation Capacity] \label{re:com_cap}
From \eqref{eq:r_u_ptp} and \eqref{eq:r_l_ptp}, the gap between $R^{(u)}_{\sf ptp}$ and $R^{(l)}_{\sf ptp}$ is zero if the condition $\bar{C}_{\sf ptp}([1:K])=\min_{i\in [1:K]}\bar{C}_{\sf ptp}(\{i\})$ is satisfied, which characterizes the computation capacity. Figure \ref{figs:example_network1} is an example of this class of networks.
Basically, any layered network with equal channel gains can be an example.
\end{remark}

\begin{remark}[Bit-Pipe Wired Networks]
As mentioned in Remark \ref{re:wired}, we can easily interpret the results in Theorems \ref{thm:com_rate_orthogonal_unequal} and \ref{thm:orthogonal_net} for bit-pipe wired networks. Hence, Remark \ref{re:com_cap} also provides the computation capacity for a certain class of bit-pipe wired networks, which closes the gap between the lower and upper bounds in \cite{Appuswamy2:11,Appuswamy:11} in the case of the arithmetic sum function computation.
\end{remark}

\subsection{Gaussian Networks With Multiple-Access} \label{subsec:result2}
We extend the idea presented in Section \ref{subsec:outline} to the computation of multiple weighted arithmetic sums, which contains the type function, and also to a general network topology in the following theorem.
For the achievability, we first abstract each multiple-access component by compute-and-forward and then apply linear network coding at each relay node to convert the original Gaussian network with multiple-access into the end-to-end modulo-$q$ sum channel.

\begin{theorem}[Gaussian Networks With Multiple-Access] \label{thm:net_multiple_acess}
Consider the Gaussian network with multiple-access. 
Then the computation rate 
\begin{equation} \label{eq:com_rate_int}
R_{\sf mac}=\frac{\min_{i\in [1:K]}C^+_{\sf mac}(\{i\})}{H(f(\mathbf{S}))}
\end{equation} 
is achievable, where
\begin{equation}
C^+_{\sf mac}(\{i\}) = \min _{\Omega\in \Lambda(\{i\})}\sum_{v\in\Omega^c}\mathbf{1}_{\Gamma_{in}(v)\cap \Omega\neq\emptyset}{\sf C}^{+}\left(\frac{1}{|\Gamma_{in}(v)|}+\min_{u\in\Gamma_{in}(v)}h_{u,v}^2P\right).
\end{equation}
Here $f(\mathbf{S})=(U_1,\cdots,U_L)$ for the arithmetic sum function and $f(\mathbf{S})=(B_0,\cdots,B_{p-1})$ for the type function.
\end{theorem}
\begin{proof}
We refer to Section \ref{subsec:net_multiple_acess} for the proof.
\end{proof}

Similarly, denote the cut-set upper bound in \eqref{eq:cutset_int} for $\Sigma=[1:K]$ as 
\begin{equation} 
R^{(u)}_{\sf mac}:=\frac{\bar{C}_{\sf mac}([1:K])}{H(f(\mathbf{S}))}
\end{equation}
and the achievable computation rate in Theorem \ref{thm:net_multiple_acess} as 
\begin{equation} 
R^{(l)}_{\sf mac}:=\frac{\min_{i\in [1:K]}C^+_{\sf mac}(\{i\})}{H(f(\mathbf{S}))}.
\end{equation}
Then the following corollary holds.

\begin{corollary}[Approximate Computation Capacity] \label{co:diff_gap_int}
If $h_{u,v}$ are the same for all $u\in \Gamma_{in}(v)$ and $\bar{C}_{\sf mac}([1:K])-\min_{i\in[1:K]}\bar{C}_{\sf mac}(\{i\})\leq c_1 |V|\log |V|$ for some constant $c_1>0$, then 
\begin{equation} \label{eq:diff_gap_int}
R^{(u)}_{\sf mac}-R^{(l)}_{\sf mac}\leq \frac{c_2|V| \log |V|}{{H(f(\mathbf{S}))}}
\end{equation}
for any power $P$, where $c_2>0$ is some constant.
\end{corollary}
\begin{proof}
We refer to \ref{subsec:net_multiple_acess} for the proof.
\end{proof}

\begin{figure}[t!]
\begin{center}
\includegraphics[scale=1.3]{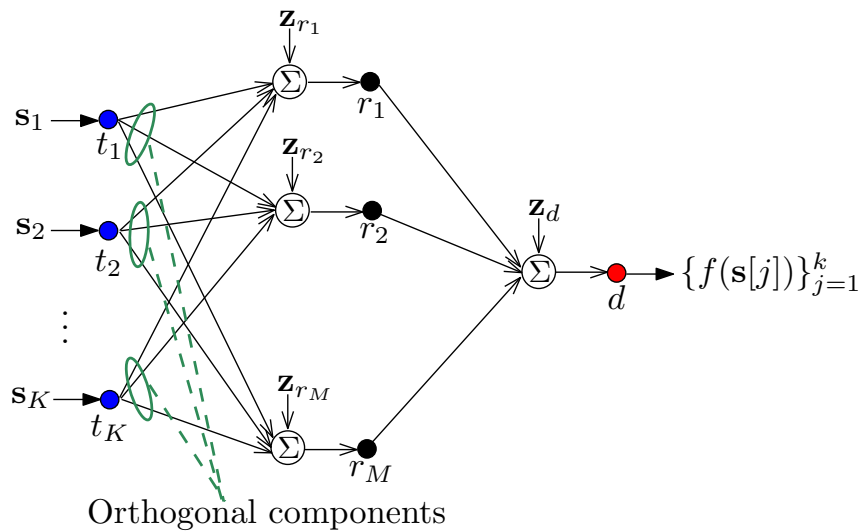}
\end{center}
\vspace{-0.15in}
\caption{An example of Gaussian networks with multiple-access that satisfies the condition in Corollary \ref{co:diff_gap_int}, where the channel coefficients are equal to one on all links.}
\label{figs:example_network2}
\vspace{-0.1in}
\end{figure}

\begin{remark}[Bounds on the Capacity Gap]
First of all, the gap in Corollary \ref{co:diff_gap_int} does not depend on $P$ and therefore provides a universal performance guarantee for any $P$.
Also, since $H(f(\mathbf{S})))$ is an increasing function of $K\leq |V|$, the gap in Corollary \ref{co:diff_gap_int} increases at most as $|V| \log |V|$.
Figure \ref{figs:example_network2} is an example of the class of networks satisfying the condition in Corollary \ref{co:diff_gap_int}.
Basically, any layered network with equal channel gains can be an example.
\end{remark}

\begin{remark}[Tighter Bound for the Gaussian MAC] \label{re:tighter_bound}
For the single-hop case, i.e., the Gaussian MAC, we can easily tighten the gap in Corollary \ref{co:diff_gap_int}. Specifically, if ${\sf C}\left((\sum_{i=1}^K h_i)^2P\right)-{\sf C}\left(\frac{1}{K}+\min_{i\in[1:K]}h_i^2P\right)\leq c_3 \log K$ for some constant $c_3>0$, then 
\begin{equation}  \label{eq:gap_mac}
R^{(u)}_{\sf mac}-R^{(l)}_{\sf mac}\leq \frac{c_4 \log K}{H(f(\mathbf{S}))}
\end{equation}
for any power $P$, where $c_4>0$ is some constant.
\end{remark}

\section{Computation Over Linear Finite Field Networks} \label{sec:finite_field}
In this section, we introduce the following linear finite field network and show how to compute the desired function over the considered network.
Specifically, we first explain the computation over the modulo-$q$ sum channel in Section \ref{subsec:com_mod_sum}.
Then we introduce the network transformation method that converts a general linear finite field network into  the modulo-$q$ sum channel in Section \ref{subsec:network_trans}.
The computation coding and transformation method presented in this section will be used for proving the results for general Gaussian networks with orthogonal components in Section \ref{sec:com_net}.

Again, we assume a network represented by a directed graph $G=(V,E)$ and the same source and desired function in Definitions \ref{def:desired_function1} to \ref{def:rv2}.
For the considered linear finite field network model, the input--output at time $t$ is given by
\begin{equation} \label{eq:input_output_field}
y^{(t)}_v=\bigoplus_{u\in\Gamma_{in}(v)}h_{u,v}x^{(t)}_u,
\end{equation}
where $x^{(t)}_u\in\mathbb{F}^{\alpha_u}_q$, $y^{(t)}_v\in\mathbb{F}^{\beta_v}_q$, and $h_{u,v}\in\mathbb{F}^{\beta_v\times \alpha_u}_q$. 
Here $q$ is assumed to be a prime number.
Then the length-$n$ time-extended input--output is represented as
\begin{equation} \label{eq:input_output_field2}
\mathbf{y}_v=\bigoplus_{u\in\Gamma_{in}(v)}\mathbf{H}_{u,v}\mathbf{x}_u,
\end{equation}
where $\mathbf{x}_u=\Big[{x^{(1)}_u}^T,\cdots,{x^{(n)}_u}^T\Big]^T\in\mathbb{F}_q^{n\alpha_u}$, $\mathbf{y}_v=\Big[{y^{(1)}_u}^T,\cdots,{y^{(n)}_u}^T\Big]^T\in\mathbb{F}_q^{n\beta_v}$, and $\mathbf{H}_{u,v}\in\mathbb{F}_q^{n\beta_v\times n\alpha_u}$ denotes the block diagonal matrix consisting of $h_{u,v}$ at each block diagonal element.

\begin{remark}
Without loss of generality, we can assume that $\alpha_v=\beta_v$ is the same for all $v\in V$ since the effect of different values of $\alpha_v$ and $\beta_v$ can be equivalently reflected by the channel matrix $h_{u,v}$.
However, we allow different values of $\alpha_v$ and $\beta_v$ in this section for easy explanation of the transformation from the Gaussian network model in Section \ref{sec:com_net}.
\end{remark}
\begin{remark}
The considered linear finite field network includes the bit-pipe wired channel model and the linear finite field deterministic model proposed in \cite{Avestimehr:11}.
\end{remark}

\subsection{Computation Over the Modulo-$q$ Sum Channel} \label{subsec:com_mod_sum}
As a special case of the considered linear finite field network, first consider the modulo-$q$ sum channel.
The following lemma shows how to compute the arithmetic sum or type function over the modulo-$q$ sum channel when the field size $q$ is large enough.
This lemma is of crucially importance to prove the main theorems in the paper. 
The key observation is that utilizing a small portion of finite field elements and then computing the corresponding modulo-$q$ sum can attain the desired function. Furthermore, linear source coding in Lemma \ref{thm:csiszar} can compensate the inefficiency of utilizing only a small portion of finite field elements and, as a result, achieves the optimal computation rate when the field size $q$ is large enough.

\begin{lemma}[Computation Over the Modulo-$q$ Sum Channel]\label{lemma:sum_modulo}
Consider the computation over the modulo-$q$ sum channel in which the length-$n$ modulo-$q$ sum channel is given as
\begin{align} \label{eq:modulo_sum_ch}
\mathbf{y}'=\bigoplus_{i=1}^K \mathbf{x}'_i, 
\end{align}
where $\mathbf{x}'_i\in{\mathbb{F}_q}^n$ and $q$ is a prime number. 
If $q>(p-1)^2K$, then the computation capacity is given by 
\begin{equation} \label{eq:sum_modulo_sum_ch}
R=\frac{\log q}{H(f(\mathbf{S}))},
\end{equation}
where  $f(\mathbf{S})=(U_1,\cdots,U_L)$ for the arithmetic sum function and $f(\mathbf{S})=(B_0,\cdots,B_{p-1})$ for the type function.

\end{lemma}
\begin{proof}
Let $g(\cdot)$ denote the mapping from a subset of integers $[0:q-1]$ to the corresponding finite field $\mathbb{F}_q$.

First consider the arithmetic sum computation.
Define $s'_i[j]=g(s_i[j])$, $a'_{li}=g(a_{li})$, and $U_l'=\bigoplus_{i=1}^Ka'_{li}g(S_i)$ for $i\in[1:K]$ and $l\in[1:L]$.
Suppose that the $i$th sender observes $\mathbf{s}'_i=[s'_i[1],\cdots s'_i[k]]^T$, which can be obtained from $\mathbf{s}_i$, and the receiver wishes to compute $(\mathbf{u}_1',\cdots,\mathbf{u}_L')=\left(\bigoplus_{i=1}^K a'_{1i}\mathbf{s}_i',\cdots,\bigoplus_{i=1}^K a'_{Li}\mathbf{s}_i'\right)$.
From Theorem \ref{thm_comp_mac}, the receiver can compute $(\mathbf{u}_1',\cdots,\mathbf{u}_L')$ reliably for $n$ sufficiently large if
\begin{equation}
\frac{k}{n}\leq \frac{\log q}{H(U'_1,\cdots,U'_L)}.
\end{equation}
Hence by setting $k=\frac{n\log q}{H(U'_1,\cdots,U'_L)}$, the computation rate 
\begin{align} \label{eq:compuration_rate_ex}
R&=\frac{k}{n}\nonumber\\
&=\frac{\log q}{H(U'_1,\cdots,U'_L)}
\end{align}
is achievable for the desired function $(\mathbf{u}_1',\cdots,\mathbf{u}_L')$.
From the condition $q>(p-1)^2K$, we have
\begin{align}
\mathbf{u}'_l&=\left[g\left(\sum_{i=1}^Ka_{li}s_i[1]\right),\cdots,g\left(\sum_{i=1}^Ka_{li}s_i[k]\right)\right]^T,\label{eq:condition11}\\
U_l'&=g\left(\sum_{i=1}^Ka_{li}S_i\right)=g(U_l).\label{eq:condition22}
\end{align}
Since $g(\cdot)$ has one-to-one correspondence, from \eqref{eq:condition11}, the receiver can compute $\left(\sum_{i=1}^Ka_{1i}\mathbf{s}_i,\cdots,\sum_{i=1}^Ka_{Li}\mathbf{s}_i\right)$ from $(\mathbf{u}_1',\cdots,\mathbf{u}_L')$,  
Finally, from the fact that $H((U'_1,\cdots,U'_L)=H(U_1,\cdots,U_L)$, which can be verified from \eqref{eq:condition22},  the achievable computation rate for $\left(\sum_{i=1}^Ka_{1i}\mathbf{s}_i,\cdots,\sum_{i=1}^Ka_{Li}\mathbf{s}_i\right)$  is given by \eqref{eq:sum_modulo_sum_ch}. 

Now consider the type computation.
Define $b_{li}[j]=\mathbf{1}_{s_i[j]=l}$.
Let $b'_{li}[j]=g(b_{li}[j])$ and $\mathbf{b}'_{li}=[b'_{li}[1],\cdots,b'_{li}[k]]^T$.
Suppose that the $i$th sender observes $\mathbf{b}'_{li}$, which can be obtained from $\mathbf{s}_i$, and the receiver wishes to compute $\{\mathbf{u}''_l=\bigoplus_{i=1}^K \mathbf{b}'_{li}\}_{l=0}^{p-1}$.
From Theorem \ref{thm_comp_mac}, the receiver can compute $\{\mathbf{u}_l''\}_{l=1}^L$ reliably for $n$ sufficiently large if 
\begin{equation}
\frac{k}{n}\leq \frac{\log q}{H(B'_0,\cdots,B'_{p-1})},
\end{equation}
where  $B_l'=\bigoplus_{i=1}^K B'_{li}$ and $B'_{li}=g(\mathbf{1}_{S_i=l})$.
Hence by setting $k=\frac{n\log q}{H(B'_0,\cdots,B'_{p-1})}$, the computation rate 
\begin{align}
R&=\frac{k}{n}\nonumber\\
&=\frac{\log q}{H(B'_0,\cdots,B'_{p-1})}
\end{align}
is achievable for the desired function $\{\mathbf{u}_l''\}_{l=0}^{p-1}$.
From the condition $q>(p-1)^2K$, we have
\begin{align}
\mathbf{u}''_l&=\left[g\left(\sum_{i=1}^Kb_{li}[1]\right),\cdots,g\left(\sum_{i=1}^Kb_{li}[k]\right)\right]^T,\nonumber\\
B'_l&=g\left(\sum_{i=1}^K \mathbf{1}_{S_i=l}\right)=g(B_l).
\end{align}
Since $g(\cdot)$ has one-to-one correspondence, the receiver can compute $\{b_0(\mathbf{s}[j]),\cdots,b_{p-1}(\mathbf{s}[j])\}_{j=1}^k$ from $\{\mathbf{u}_l''\}_{l=0}^{p-1}$.
Finally, from the fact that $H(B'_0,\cdots,B'_{p-1})=H(B_0,\cdots,B_{p-1})$, the achievable computation rate for the type function is given by \eqref{eq:sum_modulo_sum_ch}. 

The converse for both cases can be easily shown from the same cut-set argument in Section \ref{subset:cut_set}, which completes the proof.
\end{proof}

\begin{remark}[Computation of Multiple Arithmetic Sums]
Notice that the channel model in \eqref{eq:modulo_sum_ch} is exactly the same as in \eqref{eq:modulo_sum_ch11} except that $q$ is fixed in this section.
Hence Lemma \ref{lemma:sum_modulo} extends the computation of a single arithmetic sum in Section \ref{subsec:outline} to the computation of multiple weighted arithmetic sums (the type function can be represented as multiple arithmetic sums, see Definition \ref{def:desired_function2} and Remark \ref{re:sum_type}).
\end{remark}

\subsection{Computation Via Network Transformation} \label{subsec:network_trans}
We now describe how to convert a general linear finite field network into the modulo-$q$ sum channel in \eqref{eq:modulo_sum_ch} in the following lemma.  
The basic principle is similar to those proposed in \cite{Zhan:13} in the sense that linear network coding is applied at each relay node to construct an end-to-end linear finite field channel and then precoder at each sender to convert the end-to-end channel  into the linear finite field channel in \eqref{eq:modulo_sum_ch}. 

To describe the following lemma, for $\Sigma\subseteq [1:K]$, let $H_{\Omega}(\Sigma)$ denote the transfer matrix associated with the cut $\Omega\in \Lambda(\Sigma)$ on $G(\Sigma)$.

\begin{lemma}[Transform Into the Modulo-$q$ Sum Channel] \label{thm:transform}
Consider the linear finite field network in which the length-$n$ time extended input--output is given by \eqref{eq:input_output_field2} . For $n$ sufficiently large, the network can be transformed into the following modulo-$q$ sum channel: 
\begin{equation}
\mathbf{y}'_d=\bigoplus_{i=1}^K \mathbf{x}'_{t_i},
\end{equation}
where $\mathbf{x}'_{t_i}\in\mathbb{F}^{n\min_{i\in [1:K]}\min _{\Omega\in \Lambda(\{i\})}\operatorname{rank}\left(H_{\Omega}(\{i\})\right)}_q$ for all $i\in[1:K]$.
\end{lemma}
\begin{proof}
We refer to Appendix II for the proof.
\end{proof}

Based on Lemmas \ref{lemma:sum_modulo} and \ref{thm:transform}, we characterize the computation capacity when $q>(p-1)^2K$ in the following theorem.

\begin{theorem}[Linear Finite Field Networks] \label{thm:com_finite_field_net}
Consider the linear finite field network in which the length-$n$ time-extended input--output is given as in \eqref{eq:input_output_field2}. If $q>(p-1)^2K$ and
\begin{equation} \label{eq:cap_condition}
\min _{\Omega\in \Lambda([1:K])}\operatorname{rank}\left(H_{\Omega}([1:K])\right)=\min_{i\in [1:K]}\min _{\Omega\in \Lambda(\{i\})}\operatorname{rank}\left(H_{\Omega}(\{i\})\right),
\end{equation}
then the computation capacity is given by
\begin{equation} \label{eq:achievable_finite}
R=\frac{\min_{i\in [1:K]}\min _{\Omega\in \Lambda(\{i\})}\operatorname{rank}\left(H_{\Omega}(\{i\})\right)}{H(f(\mathbf{S}))},
\end{equation}
where  $f(\mathbf{S})=(U_1,\cdots,U_L)$ for the arithmetic sum function and $f(\mathbf{S})=(B_0,\cdots,B_{p-1})$ for the type function.
\end{theorem}
\begin{proof}
The achievability follows by the network transformation using Lemma \ref{thm:transform} and then the computation over the transformed modulo-$q$ sum channel using Lemma \ref{lemma:sum_modulo}.
The converse follows by the cut-set upper bound showing that $R\leq\frac{\min _{\Omega\in \Lambda([1:K])}\operatorname{rank}\left(H_{\Omega}([1:K])\right)}{H(f(\mathbf{S}))}$ and the condition in \eqref{eq:cap_condition}, which completes the proof.
\end{proof}

\begin{remark}[Type Computation]
For the type computation, the condition $q>K$ is enough for proving Lemma \ref{lemma:sum_modulo} and Theorem \ref{thm:com_finite_field_net}. 
\end{remark}

\begin{remark}[Computation Over Gaussian Networks]
We will apply the same computation coding and network transformation in Lemmas \ref{lemma:sum_modulo} and \ref{thm:transform} after converting Gaussian networks into linear finite field networks in Section \ref{sec:com_net}.
For this case, we set $q$ arbitrarily large and, as a result, the condition $q>(p-1)^2K$ in Lemma \ref{lemma:sum_modulo} and Theorem \ref{thm:com_finite_field_net} disappears for Gaussian networks.
\end{remark}

\section{Computation Over Gaussain Networks With Orthogonal Components} \label{sec:com_net}
We are now ready to prove Theorems \ref{thm:orthogonal_net} and \ref{thm:net_multiple_acess} and Corollary \ref{co:diff_gap_int}, which are about the computation over general Gaussian networks with orthogonal components.

\subsection{Orthogonal Gaussian Networks} \label{subsec:ortho_com_networks}
Consider orthogonal Gaussian networks in which the length-$n$ time-extended input--output is given as in \eqref{eq:in_out_ptp}.
In the following, we prove the achievable computation rate in Theorem \ref{thm:orthogonal_net}.
The achievability follows the network abstraction based on capacity-achieving point-to-point channel codes and then transformation into the modulo-$q$ sum channel via linear network coding at each relay (the same transformation used in Lemma \ref{thm:transform}).
Finally, we apply linear source coding for computation over the transformed modulo-$q$ sum channel as the same manner used in Lemma \ref{lemma:sum_modulo}.

\subsubsection{Proof of Theorem \ref{thm:orthogonal_net}}
By applying capacity-achieving point-to-point channel codes, for $n$ sufficiently large, the length-$n$ time extended orthogonal Gaussian channel in \eqref{eq:in_out_ptp} can be transformed into
\begin{equation} \label{eq:transform_ptp}
\mathbf{y}'_{u,v}=\mathbf{x}'_{u,v}
\end{equation}
for all $(u,v)\in E$, where $\mathbf{x}_{u,v}\in\mathbb{F}^{m_{u,v}}_q$.
Here $m_{u,v}\log q\le n{\sf C}(h^2_{u,v}P)$ should be satisfied.
Hence we set $q=K+1$ and $m_{u,v}=n{\sf C}(h^2_{u,v}P)(\log q)^{-1}$.

In order to use Theorem \ref{thm:transform}, we define $\mathbf{x}'_u=\{\mathbf{x}_{u,v}\}_{v\in\Gamma_{out}(u)}$ and $\mathbf{y}'_v=\{\mathbf{y}_{u,v}\}_{u\in\Gamma_{in}(v)}$ and represent the input--output as
\begin{equation} \label{eq:transform_ptp2}
\mathbf{y}'_v=\bigoplus_{u\in \Gamma_{in}(v)}\mathbf{H}'_{u,v}\mathbf{x}'_u.
\end{equation}
Here $\mathbf{H}'_{u,v}$ is determined by the original channel condition in \eqref{eq:transform_ptp}.
Notice that the input--output in \eqref{eq:transform_ptp2} is the same linear finite field model considered in Section \ref{sec:finite_field}, see \eqref{eq:input_output_field}.
Then we can treat $\mathbf{y}'_v$ and $\mathbf{x}'_u$ as `super-symbol' and apply the multi-letter coding over these length-$\eta$ super-symbols, that is, length-$\eta n$ symbols.
Hence, from Lemma \ref{thm:transform}, we can transform the linear finite field network with the input--output \eqref{eq:transform_ptp2} into the following modulo-$q$ sum channel: 
\begin{equation}
\mathbf{y}'_d=\bigoplus_{i=1}^K \mathbf{x}'_{t_i},
\end{equation}
where $\mathbf{x}'_{t_i}\in\mathbb{F}^{\eta\min_{i\in [1:K]}\min _{\Omega\in \Lambda(\{i\})}\operatorname{rank}\left(H_{\Omega}(\{i\})\right)}_q$ for all $i\in[1:K]$.

Now consider the computation over the above modulo-$q$ sum channel. 
Since $q=K+1>K$, from Lemma \ref{lemma:sum_modulo}, the receiver can compute the desired function reliably for $\eta$ sufficiently large if 
\begin{equation}
\frac{k}{\eta\min_{i\in [1:K]}\min _{\Omega\in \Lambda(\{i\})}\operatorname{rank}\left(H_{\Omega}(\{i\})\right)}\leq \frac{\log q}{H(f(\mathbf{S}))}.
\end{equation}
Hence by setting $k=\frac{\eta\min_{i\in [1:K]}\min _{\Omega\in \Lambda(\{i\})}\operatorname{rank}\left(H_{\Omega}(\{i\})\right)\log q}{H(f(\mathbf{S}))}$, the computation rate 
\begin{align} \label{eq:compuration_rate}
R&=\frac{k}{\eta n}\nonumber\\
&=\frac{\min_{i\in [1:K]}\min _{\Omega\in \Lambda(\{i\})}\operatorname{rank}\left(H_{\Omega}(\{i\})\right)\log q}{nH(f(\mathbf{S}))}\nonumber\\
&\overset{(a)}{=}\frac{\min_{i\in [1:K]}\min _{\Omega\in \Lambda(\{i\})}\sum_{(u,v)\in E(\{i\}),u\in \Omega,v\in\Omega^c}m_{u,v}\log q}{nH(f(\mathbf{S}))}\nonumber\\
&\overset{(b)}{=}\frac{\min_{i\in [1:K]}\min _{\Omega\in \Lambda(\{i\})}\sum_{(u,v)\in E(\{i\}),u\in \Omega,v\in\Omega^c}{\sf C}(h^2_{u,v}P)}{H(f(\mathbf{S}))}\nonumber\\
&\overset{(c)}{=}\frac{\min_{i\in [1:K]}\bar{C}_{\sf ptp}(\{i\})}{H(f(\mathbf{S}))}
\end{align}
is achievable, where $(a)$ follows since $\operatorname{rank}\left(H_{\Omega}(\{i\})\right)$ assuming the channel \eqref{eq:transform_ptp2} is the same as that assuming the channel \eqref{eq:transform_ptp} and $\operatorname{rank}\left(H_{\Omega}(\{i\})\right)=\sum_{(u,v)\in E(\{i\}),u\in \Omega,v\in\Omega^c}m_{u,v}$, $(b)$ follows since $m_{u,v}=n{\sf C}(h^2_{u,v}P)(\log q)^{-1}$, and $(c)$ follows from the definition \eqref{eq:min_cut_ptp}. 
In conclusion, Theorem \ref{thm:orthogonal_net} holds.

\subsection{Gaussian Networks With Multiple-Access} \label{subsec:net_multiple_acess}
Consider Gaussian networks with multiple-access in which the length-$n$ time-extended input--output is given as in \eqref{eq:in_out}.
In the following, we prove the achievable computation rate in Theorem \ref{thm:net_multiple_acess} and Corollary \ref{co:diff_gap_int}.
The achievability follows the network abstraction based on compute-and-forward in Theorem \ref{thm:compute-forward} and the rest of the procedure is similar to that in Section \ref{subsec:ortho_com_networks}.

\subsubsection{Proof of Theorem \ref{thm:net_multiple_acess}}
For each node $v\in V$, suppose that node $u\in \Gamma_{in}(v)$ observes $\mathbf{x}'_{u,v}\in\mathbb{F}^{m_v}_q$ and node $v$ wishes to decode $\bigoplus_{u\in \Gamma_{in}(v)}\mathbf{x}'_{u,v}$.
Let $\mathbf{x}^{\sf lattice}_{u,v}(\mathbf{x}'_{u,v})$ denote a dither-added transmit lattice point from node $u$ to node $v$ for the compute-and-forward framework in \cite{Nazer:11}, which satisfies the power constraint $P$.
Then node $u$ transmits $\mathbf{x}_{u,v}(\mathbf{x}'_{u,v})=\frac{\min_{u\in \Gamma_{in}(v)}\{|h_{u,v}|\}}{h_{u,v}}\mathbf{x}^{\sf lattice}_i(\mathbf{x}'_i)$ to node $v$.
We can equivalently interpret that $h_{u,v}=1$ and the average power constraint from node $u$ to node $v$ is given as $\min_{u\in \Gamma_{in}(v)}\{h_{u,v}^2\}P$.
From Theorem \ref{thm:compute-forward} (see also Example \ref{ex:compute_forward_equal_ch}), node $v$ can decode $\bigoplus_{u\in \Gamma_{in}(v)}\mathbf{x}'_{u,v}$ reliably for $n$ sufficiently large if
$m_v\leq n{\sf C}^{+}\left(\frac{1}{|\Gamma_{in}(v)|}+\min_{u\in\Gamma_{in}(v)}h_{u,v}^2P\right)(\log q)^{-1}$
and $q$ is an increasing function of $n$ satisfying that $q\to\infty$ as $n\to \infty$.
Therefore, by treating $\mathbf{x}'_{u,v}$ as the channel input from node $u$ to node $v$ and $\bigoplus_{u\in\Gamma_{in}(v)}\mathbf{x}'_{u,v}$ as the channel output of node $v$, for $n$ sufficiently large, we can transform each length-$n$ Gaussian multiple-access component \eqref{eq:in_out} into the following length-$m_v$ modulo-$q$ sum channel:
\begin{equation} \label{eq:in_out_int3}
\mathbf{y}'_{v}=\bigoplus_{u\in \Gamma_{in}(v)}\mathbf{x}'_{u,v},
\end{equation}
where $\mathbf{x}'_{u,v}\in \mathbb{F}_q^{m_{v}}$. 
Here $q$ is  the largest prime number among $[1:n\log n]$ and 
\begin{equation} \label{eq:m_v_int}
m_v=n{\sf C}^{+}\left(\frac{1}{|\Gamma_{in}(v)|}+\min_{u\in\Gamma_{in}(v)}h_{u,v}^2P\right)(\log q)^{-1}.
\end{equation}
We again define $\mathbf{x}'_u=\{\mathbf{x}_{u,v}\}_{v\in\Gamma_{out}(u)}$ and represent the input--output as
\begin{equation} \label{eq:in_out_int4}
\mathbf{y}'_v=\sum_{u\in \Gamma_{in}(v)}\mathbf{H}'_{u,v}\mathbf{x}'_u.
\end{equation}
Here $\mathbf{H}'_{u,v}$ is determined by the original channel condition in \eqref{eq:in_out_int3}.
Then, as the same manner from \eqref{eq:transform_ptp2} to \eqref{eq:compuration_rate}, we can apply multi-letter coding and, from Lemmas \ref{lemma:sum_modulo} and \ref{thm:transform}, the computation rate 
\begin{align}
R&=\frac{k}{\eta n}\nonumber\\
&=\frac{\min_{i\in [1:K]}\min _{\Omega\in \Lambda(\{i\})}\operatorname{rank}\left(H_{\Omega}(\{i\})\right)\log q}{nH(f(\mathbf{S}))}\nonumber\\
&\overset{(a)}{=}\frac{\min_{i\in [1:K]}\min _{\Omega\in \Lambda(\{i\})}\sum_{v\in\Omega^c}\mathbf{1}_{\Gamma_{in}(v)\cap \Omega\neq\emptyset}m_v\log q}{nH(f(\mathbf{S}))}\nonumber\\
&\overset{(b)}{=}\frac{\min_{i\in [1:K]}\min _{\Omega\in \Lambda(\{i\})}\sum_{v\in\Omega^c}\mathbf{1}_{\Gamma_{in}(v)\cap \Omega\neq\emptyset}{\sf C}^{+}\left(\frac{1}{|\Gamma_{in}(v)|}+\min_{u\in\Gamma_{in}(v)}h_{u,v}^2P\right)}{H(f(\mathbf{S}))}
\end{align}
is achievable, where $(a)$ follows since $\operatorname{rank}\left(H_{\Omega}(\{i\})\right)$ assuming the channel \eqref{eq:in_out_int4} is the same as that assuming the channel \eqref{eq:in_out_int3} and $\operatorname{rank}\left(H_{\Omega}(\{i\})\right)=\sum_{v\in\Omega^c}\mathbf{1}_{\Gamma_{in}(v)\cap \Omega\neq\emptyset}m_v$ and $(b)$ follows from \eqref{eq:m_v_int}.
Note that $q$, which is the largest prime number among $[1:n\log n]$, becomes arbitrarily large as $n\to\infty$ and, as a result, satisfies the condition $q>(p-1)^2K$ in Lemma \ref{lemma:sum_modulo} for $n$ sufficiently large.
In conclusion, Theorem \ref{thm:net_multiple_acess} holds.

\subsubsection{Proof of Corollary \ref{co:diff_gap_int}}
Denote  
\begin{equation}
i^*={\arg\min}_{i\in[1:K]}C^+_{\sf mac}(\{i\})
\end{equation}
and
\begin{equation}
\Omega^*={\arg\min} _{\Omega\in \Lambda(\{i^*\})}\sum_{v\in\Omega^c}\mathbf{1}_{\Gamma_{in}(v)\cap \Omega\neq\emptyset}{\sf C}^{+}\left(\frac{1}{|\Gamma_{in}(v)|}+\min_{u\in\Gamma_{in}(v)}h_{u,v}^2P\right).
\end{equation}
Then
\begin{align}
&H(f(\mathbf{S}))(R^{(u)}_{\sf mac}-R^{(l)}_{\sf mac})\nonumber\\
&=\bar{C}_{\sf mac}([1:K])-\min_{i\in [1:K]}C^+_{\sf mac}(\{i\})\nonumber\\
&\overset{(a)}{\leq}\bar{C}_{\sf mac}(\{i^*\})-C^+_{\sf mac}(\{i^*\})+\left(\bar{C}_{\sf mac}([1:K])-\min_{i\in[1:K]}\bar{C}_{\sf mac}(\{i\})\right)\nonumber\\
&\overset{(b)}{\leq}\sum_{v\in \Omega^{*c}}{\sf C}\left(\Big(\sum_{u\in \Gamma_{in}(v), u\in \Omega^*}h_{u,v}\Big)^2 P\right)-\sum_{v\in\Omega^{*c}}\mathbf{1}_{\Gamma_{in}(v)\cap \Omega^*\neq\emptyset}{\sf C}^{+}\left(\frac{1}{|\Gamma_{in}(v)|}+\min_{u\in\Gamma_{in}(v)}h_{u,v}^2P\right)\nonumber\\
&{~~~}+\left(\bar{C}_{\sf mac}([1:K])-\min_{i\in[1:K]}\bar{C}_{\sf mac}(\{i\})\right)\nonumber\\
&\overset{(c)}{\leq} \sum_{v\in \Omega^{*c}}\mathbf{1}_{\Gamma_{in}(v)\cap \Omega^*\neq\emptyset}\left({\sf C}(|V|^2h_{u,v}^2P)-{\sf C}^+\left(h_{u,v}^2 P\right)\right)\nonumber\\
&{~~~}+\left(\bar{C}_{\sf mac}([1:K])-\min_{i\in[1:K]}\bar{C}_{\sf mac}(\{i\})\right)\nonumber\\
&\overset{(d)}{\leq} |V| (\log (|V|^2)+1)+\left(\bar{C}_{\sf mac}([1:K])-\min_{i\in[1:K]}\bar{C}_{\sf mac}(\{i\})\right),
\end{align}
where $(a)$ follows from the definition of $i^*$, $(b)$ follows from the definition of $\Omega^*$, $(c)$ follows from the condition that $h_{u,v}$ are the same for all $u\in \Gamma_{in}(v)$, and $(d)$ follows since ${\sf  C}(x)-{\sf C}^+(x)\leq 1$ for all $x\ge 0$.
Finally, from the condition $\bar{C}_{\sf mac}([1:K])-\min_{i\in[1:K]}\bar{C}_{\sf mac}(\{i\})\leq c_1 |V|\log |V|$, we have \eqref{eq:diff_gap_int}, which completes the proof.

\section{Extensions}

In this section, we apply our computation code to other interesting scenarios for computing over Gaussian networks.

\subsection{Modulo-$p$ Sum Computation}
Consider the modulo-$p$ sum computation over Gaussian networks with orthogonal components. 
In \cite{Nazer:07,Zhan:13}, achievable computation rates have been derived assuming an arbitrarily large source field size, i.e., $p\to\infty$ as the block length $n$ increases.
When the source field size is fixed, it is hard to apply the previous work to find a better computation rate than the separation-based computation.
From our framework, on the other hand, one naive approach is for the fusion center to first compute the corresponding arithmetic sum over networks and then take the modulo-$p$ operation in order to obtain the desired modulo-$p$ sum.
Hence the achievable computation rates for the arithmetic sum in this paper can also be achievable computation rates for the corresponding modulo-$p$ sum.
Obviously, this approach is not optimal but we can easily find examples that it outperforms the separation-based computation.

\subsection{Superposition Approach for Unequal Channel Coefficients}
One drawback of the achievability in Theorem \ref{thm:net_multiple_acess} is when the channel coefficients have different values.
For the single-hop case, i.e., the Gaussian MAC in which the length-$n$ time-extended input--output is given by \eqref{eq:gaussian_mac}, Theorem \ref{thm:net_multiple_acess} provides 
\begin{equation} \label{eq:achi_mac}
R=\frac{{\sf C}^{+}\left(\frac{1}{K}+\min_{i\in[1:K]}h_i^2P\right)}{H(f(\mathbf{S}))}.
\end{equation}

As shown in \eqref{eq:achi_mac}, the achievable computation rate is bounded by the minimum of the channel gains.
In order to achieve computation rates scalable with $P$ from the compute-and-forward framework \cite{Nazer:11}, the transmit power of each lattice should be reduced to let the received power for each lattice be the same at the receiver side, resulting $\min_{i\in[1:K]}h_i^2$ in \eqref{eq:achi_mac}.
An improved computation rate is achievable by the superposition of multiple codes and then allocating residual transmit power to high layer codes.
The following theorem provides an improved computation rate for the two-user case.

\begin{theorem} \label{thm:superposition}
Consider the $2$-user Gaussian MAC with $h^2_2\geq h_1^2$.
Then the computation rate satisfying 
\begin{align} \label{eq:layering}
R&\le\frac{{\sf C}^{+}\left(\frac{1}{2}+h_1^2P\right)}{H(S_1|S_2)},\nonumber\\
R&\le\frac{{\sf C}^{+}\left(\frac{1}{2}+h_1^2P\right)}{H(f(\mathbf{S}))}+\frac{H(f(\mathbf{S}))-H(S_1|S_2)}{H(S_2)}\frac{{\sf C}\left(\frac{(h_2^2-h_1^2)P}{1+2h_1^2P}\right)}{H(f(\mathbf{S}))}
\end{align}
is achievable.
\end{theorem}
\begin{proof}
Let $\mathbf{x}^{\sf lattice}_i$ denote a dither-added transmit signal from lattice codebook of the $i$th sender for  compute-and-forward in \cite{Nazer:11}, which satisfies the power constraint $P$.
Let $\mathbf{x}^{\sf random}$ denote the transmit signal from a capacity-achieving point-to-point Gaussian codebook, which satisfies the power constraint $P$.
Then the first sender transmits $\mathbf{x}_1=\mathbf{x}^{\sf lattice}_1$ and the second sender transmits $\mathbf{x}_2=\frac{h_1}{h_2}\mathbf{x}^{\sf lattice}_2+\frac{\sqrt{h^2_2-h^2_1}}{h_2}\mathbf{x}^{\sf random}$.
Let $R_2$ denote the message rate delivered by $\mathbf{x}^{\sf random}$.
The receiver first decodes the message delivered by $\mathbf{x}^{\sf random}$ that is reliably decodable if $R_2\leq {\sf C}\left(\frac{(h_2^2-h_1^2)P}{1+2h_1^2P}\right)$.  The receiver then subtracts $\mathbf{x}^{\sf random}$ from the received signal. 
Then, from the same argument in the proof of Theorem \ref{thm:net_multiple_acess}, we can construct the modulo-$q$ sum channel in \eqref{eq:modulo_sum_ch} with $m=n{\sf C}^{+}\left(\frac{1}{K}+h_1^2P\right)(\log q)^{-1}$, where $q$ is  set to be the largest prime number among $[1:n\log n]$.
By utilizing the length $nR_1(\log q)^{-1}$ of this modulo-$q$ sum channel, the first sender is able to transmit its message at the rate of $R_1$ if $R_1\leq {\sf C}^{+}\left(\frac{1}{2}+h_1^2P\right)$.
The remaining modulo-$q$ sum channel has the length of $n\left({\sf C}^{+}\left(\frac{1}{2}+h_1^2P\right)-R_1\right)(\log q)^{-1}$.

For the bit-pipe channels with the rates $R_1$ and $R_2$, we apply Slepian--Wolf source coding to deliver two sources separately, then compute the desired function.
For the remaining modulo-$q$ sum channel, we  apply  linear source coding for computation in Lemma \ref{lemma:sum_modulo}.
This approach achieves the compute rate represented by the following rate constraints:
\begin{align}
R&=R'+R'',\nonumber\\
R_1&\leq {\sf C}^{+}\left(\frac{1}{2}+h_1^2P\right),\nonumber\\
R_2&\leq {\sf C}\left(\frac{(h_2^2-h_1^2)P}{1+2h_1^2P}\right),\nonumber\\
R'&\leq \frac{R_1}{H(S_1|S_2)},\nonumber\\
R'&\leq \frac{R_2}{H(S_2|S_1)},\nonumber\\
R'&\leq \frac{R_1+R_2}{H(S_1,S_2)},\nonumber\\
R''&\leq\frac{{\sf C}^{+}\left(\frac{1}{2}+h_1^2P\right)-R_1}{H(f(\mathbf{S}))}.
\end{align}
After Fourier--Motzkin elimination, we have \eqref{eq:layering}, which completes the proof.
\end{proof}

\begin{example}[Arithmetic Sum of I.I.D. Binary Sources]
Suppose that $K=2$, $h_1=1$, $h_2\geq 1$, and $S_1,S_2$ are independently and uniformly drawn from $\{0,1\}$.
The receiver wishes to compute $\{f(\mathbf{s}[j])= s_1[j]+s_2[j]\}_{j=1}^k$.
Figure \ref{figs:superposition}  plots the achievable computation rate by the superposition in Theorem \ref{thm:superposition}. As $h^2$ increases, the separation-based computation outperforms the computation scheme in Theorem \ref{thm:net_multiple_acess} and the gap to the cut-set upper bound decreases.
As shown in the figure, the superposition approach can attain both of the computation in Theorem \ref{thm:net_multiple_acess} and the separation-based computation.
\end{example}

\begin{remark}[Multiple Layering for More Than Two Users]
For more than two users, we can apply the same superposition approach by layering multiple lattice codes. We refer to \cite[Section III. F]{Zhan:13} for the detailed lattice code construction and encoding, decoding procedure. 
\end{remark}

\begin{figure}[t!]
\begin{center}
\includegraphics[scale=1]{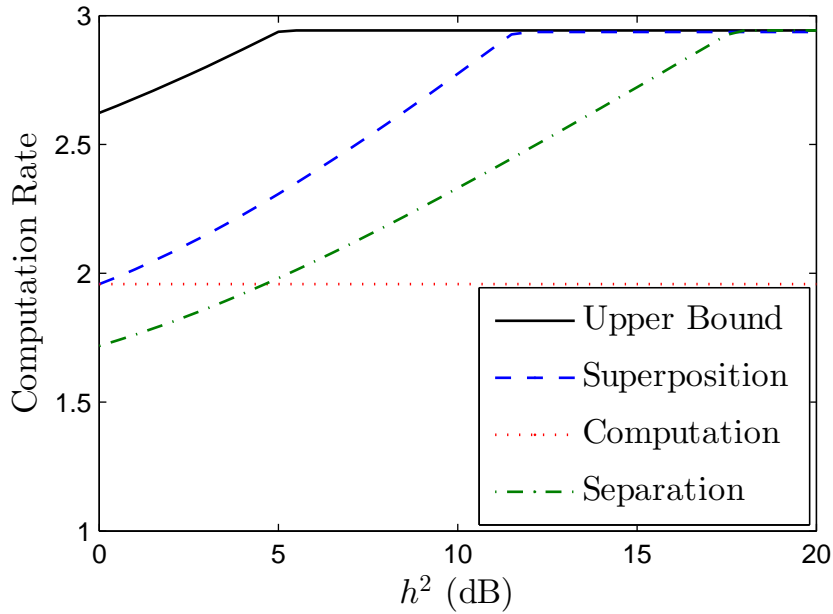}
\end{center}
\vspace{-0.15in}
\caption{Computation of $S_1+S_2$ for the two-user Gaussian MAC with unequal channel gains when $P=15$ dB.}
\label{figs:superposition}
\vspace{-0.1in}
\end{figure}

\subsection{Fading Networks}
Consider fading Gaussian MAC in which channel coefficients vary independently over time.
The length-$n$ time-extended input--output is given by
\begin{equation} 
\label{eq:fading_mac}
\mathbf{y}=\sum_{i=1}^K \mathbf{H}_i \mathbf{x}_i+\mathbf{z},
\end{equation}
where $\mathbf{H}_i=\operatorname{diag}\left({h}_i^{(1)},\cdots, h_i^{(n)}\right)$ and ${h}_i^{(t)}$ denotes the complex channel coefficient at time $t$ from the $i$th sender to the receiver.
We assume that $\{{h}_i^{(t)}\}$ are independently drawn from $\mathcal{CN}(0,1)$ and also independent over time.
We further assume that the elements of $\mathbf{z}$ are independently drawn from $\mathcal{CN}(0,1)$. 
Global channel state information is assumed to be available at each sender and the receiver.
Since $\{h_i^{(t)}\}$ are i.i.d. over time, we drop the time index hereafter for simplicity.

If we simply apply Theorem \ref{thm:net_multiple_acess} for each time slot, then due to $\min_{i\in[1:K]}|h_i|^2$ in Theorem \ref{thm:net_multiple_acess}, the achievable computation rate for fading will decrease as $K$ increases.\footnote{Although we describe the paper including Theorem \ref{thm:net_multiple_acess} based on the real channel model, the results in the paper can be straightforwardly extended to the complex channel model.} 
We can fix this problem and achieve an approximate computation capacity similar to Corollary \ref{co:diff_gap_int} and Remark \ref{re:tighter_bound} for fading Gaussian MAC.

To achieve this, we modify Theorem \ref{thm:net_multiple_acess} to the fading scenario in a similar approach as given by Goldsmith and Varaiya \cite{Goldsmith:97}. 
At each time slot, the $j$th sender transmits with power $\displaystyle\frac{\min_{i\in[1:K]}|h_i|^2P}{|h_j|^2\E[\min_{i\in[1:K]}|h_i|^2/|h_j|^2]}$, which satisfies the average power constraint.
Then from the same analysis in Theorem \ref{thm:net_multiple_acess}, the computation rate 
\begin{equation}
R=\frac{1}{H(f(\mathbf{S}))}\E\left[\C^+\left(\frac{1}{K}+\frac{\min_{i\in[1:K]}|h_i|^2P}{\E[\min_{i\in[1:K]}|h_i|^2/|h_1|^2]}\right)\right]:=R^{(l)}_{\sf fading}
\end{equation}
is achievable, where ${\sf C}^+(x):=\max\left\{\log (x),0\right\}$.
On the other hand, the cut-set upper bound is given by
\begin{equation} \label{eq:r_u_fading}
R \le \frac{1}{H(f(\mathbf{S}))}\max_{\phi_1,\cdots,\phi_K}\E\left[\C\left(\left(\sum_{i=1}^K|h_i|\phi_i(h_1,\cdots,h_K)\right)^2\right)\right]:=R^{(u)}_{\sf fading},
\end{equation}
where $\phi_i$ denotes the power allocation policy of the $i$th sender that should satisfy $\E[\phi_i^2(h_1,\cdots,h_K)]\le P$ and ${\sf C}(x):=\log (1+x)$.
The following theorem establishes an approximate computation capacity for i.i.d. Rayleigh fading.

\begin{theorem}[Approximate Computation Capacity for Fading MAC]
Consider fading Gaussian MAC with time-varying channel coefficients.
For i.i.d. Rayleigh fading, 
\begin{align}
R^{(u)}_{\sf fading}-R^{(l)}_{\sf fading}\leq \frac{3\log K+2+\log e}{H(f(\mathbf{S}))}
\end{align}
for any $P$, where $f(\mathbf{S})=(U_1,\cdots,U_L)$ for the arithmetic sum function and $f(\mathbf{S})=(B_0,\cdots,B_{p-1})$ for the type function.
\end{theorem}
\begin{proof}
Let $\{\phi_i^*\}$ denote an optimum power allocation policy maximizing \eqref{eq:r_u_fading}.
Then 
\begin{align} \label{eq:appro_bound1}
&H(f(\mathbf{S}))(R^{(u)}_{\sf fading}-R^{(l)}_{\sf fading}) \nonumber\\
&\overset{(a)}{\le} \E\left[\log\left(\frac{1+\left(\sum_{i=1}^K|h_i|\phi_i^*(h_1,\cdots,h_K)\right)^2}{1+K\frac{\min_{i\in[1:K]}|h_i|^2}{\E[\min_{i\in[1:K]}|h_i|^2/|h_1|^2]}P}\right)\right] +\log K\nonumber\\
&\le\E\left[\log\left(\frac{K\frac{\min_{i\in[1:K]}|h_i|^2}{\E[\min_{i\in[1:K]}|h_i|^2/|h_1|^2]}+\frac{1}{P}\left(\sum_{i=1}^K|h_i|\phi_i^*(h_1,\cdots,h_K)\right)^2}{K\frac{\min_{i\in[1:K]}|h_i|^2}{\E[\min_{i\in[1:K]}|h_i|^2/|h_1|^2]}}\right)\right] +\log K\nonumber\\
&\overset{(b)}{\le} 2\E\left[\log\left(\sqrt{K\frac{\min_{i\in[1:K]}|h_i|^2}{\E[\min_{i\in[1:K]}|h_i|^2/|h_1|^2]}}+\frac{1}{\sqrt{P}}\left(\sum_{i=1}^K|h_i|\phi_i^*(h_1,\cdots,h_K)\right)\right)\right] \nonumber\\
&\hspace{0.5cm}+\log\left(\E\left[\min_{i\in[1:K]}|h_i|^2/|h_1|^2\right]\right)-\E\left[\log\left(\min_{i\in[1:K]}|h_i|^2\right)\right]\nonumber\\
& \overset{(c)}{\le} 2\log\left(\sqrt{K\frac{\E\left[\min_{i\in[1:K]}|h_i|^2\right]}{\E[\min_{i\in[1:K]}|h_i|^2/|h_1|^2]}}+\frac{1}{\sqrt{P}}\left(\sum_{i=1}^K\E\left[|h_i|\phi_i^*(h_1,\cdots,h_K)\right]\right)\right) \nonumber\\
&\hspace{0.5cm}+\log\left(\E\left[\min_{i\in[1:K]}|h_i|^2/|h_1|^2\right]\right)-\E\left[\log\left(\min_{i\in[1:K]}|h_i|^2\right)\right]\nonumber\\
&\overset{(d)}{\le} 2\log\left(\sqrt{K\frac{\E\left[\min_{i\in[1:K]}|h_i|^2\right]}{\E[\min_{i\in[1:K]}|h_i|^2/|h_1|^2]}}+K\right)+\log\left(\E\left[\min_{i\in[1:K]}|h_i|^2/|h_1|^2\right]\right)-\E\left[\log\left(\min_{i\in[1:K]}|h_i|^2\right)\right]\nonumber\\
&= 2\log\left(\sqrt{K\E\left[\min_{i\in[1:K]}|h_i|^2\right]}+K\sqrt{\E\left[\min_{i\in[1:K]}|h_i|^2/|h_1|^2\right]}\right)-\E\left[\log\left(\min_{i\in[1:K]}|h_i|^2\right)\right] \nonumber\\
&\le 2\log(\sqrt{K}+K)-\E\left[\log\left(\min_{i\in[1:K]}|h_i|^2\right)\right].
\end{align}
where $(a)$ follows since ${\sf C}^+\left(\frac{1}{K}+P\right)\geq \frac{1}{2}\log(1+KP)-\frac{1}{2}\log K$, $(b)$ follows since $\log (a+b)\leq 2\log (\sqrt{a}+\sqrt{b})$ for $a\geq0$ and $b\geq0$, $(c)$ follows from Jensen's inequality, and $(d)$ follows since $\E\left[|h_i|\phi_i^*(h_1,\cdots,h_K)\right] \le \sqrt{\E[|h_i|^2]\E\left[(\phi_i^*(h_1,\cdots,h_K))^2\right]} \le \sqrt{P}$.

For Rayleigh fading, $|h_i|^2$ is exponentially distributed. Thus, we have 
\begin{align} \label{eq:appro_bound2}
\E\left[\log\left(\min_{i\in[1:K]}|h_i|^2\right)\right] &= \int_0^\infty \log(x)K\exp(-Kx) \, dx \nonumber\\
&= -\log K + \int_0^\infty \log(u)\exp(-u) \, du \nonumber\\
&\ge -\log K + \int_0^1 \log(u) \, du + \int_1^\infty \log(u)\exp(-u) \, du \nonumber\\
&\ge -\log K -\log e.
\end{align}

Therefore, from \eqref{eq:appro_bound1} and \eqref{eq:appro_bound2}, 
\begin{align}
R^{(u)}_{\sf fading}-R^{(l)}_{\sf fading}\leq \frac{3\log K+2+\log e}{H(f(\mathbf{S}))},
\end{align}
which completes the proof.
\end{proof}

\subsection{Scaling Laws}
One interesting performance metric is to focus on how the computation rate scales as the number of sources $K$ increases.
The work \cite{Giridhar:05} studied scaling laws on the computation rate under collocated collision networks assuming that concurrent transmission from multiple nodes causes collisions and, therefore, is not allowed.
It was shown in \cite{Giridhar:05} that the order of $\frac{1}{K}$ rate scaling law is achievable for the type computation, which is the same scaling law achievable by the separation-based computation.
Whereas, Theorem \ref{thm:net_multiple_acess} provides the order of $\frac{1}{\log K}$ rate scaling law for collocated Gaussian networks.
The gain comes from  a more efficient physical layer abstraction using the compute-and-forward framework and exploiting the superposition property of the abstracted channel  for function computation, which is not allowed for collocated collision networks.

\subsection{Function Multicast}
Throughout the paper, we assume that a single receiver wishes to compute the desired function.
Now consider the function multicasting problem in which multiple receivers wish to compute the same desired function. It has been shown in \cite{Suh:12} that function alignment is essentially required for linear finite field single-hop networks in order to optimally compute the modulo sum function at multiple receivers. 
For the multihop case, if there is a relay node that is connected from all senders and also at the same time connected to all receivers, this relay node can first compute the desired function and then broadcast it to multiple receivers.
For the first task, i.e., function computation at a single relay node, our coding schemes are applicable.  
For the second task, i.e., function broadcast to multiple receivers, quantize-map-and-forward \cite{Avestimehr:11} or noisy network coding \cite{Lim:11} achieves a near-optimal rate.
Although this approach is not applicable for any network topology, we can easily find examples of interest that it achieves a near-optimal computation rate. 
A similar approach has been also proposed in \cite{Kannan:13}, analyzing the rate scaling law under a bit-pipe wired network represented by an undirected graph. 

\section{Conclusion}
In this paper, we studied the function computation over Gaussian networks assuming orthogonal components.
We proposed a novel computation coding that is able to compute multiple weighted arithmetic sums including the type function.
Computing the type function is very powerful since any symmetric function such as the same mean, maximum, minimum, and so on, can be obtained from the type function.
Hence, the proposed computation coding is useful not only for the arithmetic sum computation, but for any symmetric function computation.
The main ingredients of the proposed computation coding are the network transformation via lattice codes and linear network coding and then the computation based on linear Slepian--Wolf source coding for computing.
In many cases, the proposed computation coding outperforms the separation-based computation, especially when the number of sources becomes large.
We established the computation capacity for a class of orthogonal Gaussian networks and an approximate capacity for a class of Gaussian networks with multiple-access.

\section*{Appendix I\\Hybrid Approach} \label{APP:app0}
In this appendix, we prove Theorem  \ref{thm:com_rate_orthogonal_unequal}.
First consider a distributed source coding problem with rate tuple $(R_1,\cdots,R_K)$ to compute the desired function.
The $i$th bit-pipe orthogonal channel with rate $R_i$ can be treated as the orthogonal finite field channel with $\mathbb{F}^{kR_i(\log q)^{-1}}_q$ .
From the same argument in Lemma \ref{lemma:sum_modulo}, setting $q>(p-1)^2K$ and computing the corresponding modulo-$q$ sum function yields the desired function.
We generalize the coding scheme in \cite[Section VI]{Ahlswede:83} to the $K$-user case, (Also see \cite[Theorem III.2]{Huang:12} for the rate constraints for decoding the set of auxiliary random sequences).
Then  any rate tuple $(R_1,\cdots, R_K)$ satisfying
\begin{align} \label{eq:ahlswede}
\sum_{i\in \Sigma}R_i\geq I(\{W_i\}_{i\in\Sigma};\mathbf{S}|\{W_i\}_{i\in[1:K]\setminus\Omega})+|\Sigma|H(f(\mathbf{S})|W_1,\cdots, W_K)
\end{align}
for all $\Sigma\subseteq [1:K]$ is a necessary condition for the desired function computation.
By abstracting the $i$th orthogonal Gaussian channel using point-to-point capacity-achieving codes, we have error-free bit-pipe channel with rate ${\sf C}(h_i^2)$.
Hence if a rate tuple $(\frac{{\sf C}(h_i^2)}{R},\cdots,\frac{{\sf C}(h_K^2)}{R})$ is located inside the region \eqref{eq:ahlswede}, the computation rate $R$ is achievable, which provides the rate constraint on $R$ as in \eqref{eq:com_rate_orthogonal_unequal}.
In conclusion, Theorem \ref{thm:com_rate_orthogonal_unequal} holds.

\section*{Appendix II\\Network Transformation} \label{APP:app1}
In this appendix, we prove Lemma  \ref{thm:transform}.

\subsection{Layered Networks} \label{subsec:layered}
In this subsection, we prove Lemma \ref{thm:transform} for the layered case. For layered networks, we can partition the set of nodes into $M$ layers. 
Let $V[j]\subseteq V$ denote the set of nodes at the $j$th layer, where  $j\in[1:M]$.
We assume that $V[1]$ is the set of senders and the node at the $M$th layer is the receiver.
That is, $V[1]=\{t_i\}_{i=1}^K$ and $V[M]=\{d\}$.
The encoding functions are set as follows:
\begin{itemize}
\item (Sender Encoding) Node $v\in V[1]$ transmits $\mathbf{x}_v=\mathbf{F}_v\mathbf{x}''_v$, where $\mathbf{F}_v\in\mathbb{F}_q^{n\alpha_v\times n\tau}$ and $\mathbf{x}''_v\in \mathbb{F}_q^{n \tau}$. We will specify $\tau$ later.
\item (Relay Encoding) Node $v\in\bigcup_{j=2}^{M-1}V[j]$ transmits $\mathbf{x}_v=\mathbf{F}_v\mathbf{y}_v$, where $\mathbf{F}_v\in\mathbb{F}_q^{n\alpha_v\times n\beta_v}$.
\end{itemize} 
Then the receiver generates $\mathbf{y}'_d=\mathbf{F}_d\mathbf{y}_d$, where $\mathbf{F}_d\in\mathbb{F}_q^{n\tau\times n\beta_d}$.

Let $\Gamma^{u}_{in}(v)=\{w:\mbox{there exists a direct path from $u$ to $w$}, w\in\Gamma_{in}(v)\}$.
From the definition, $\Gamma^{u}_{in}(v)\subseteq \Gamma_{in}(v)$.
Then the input--output from $\mathbf{x}'_u$, $u\in V[1]$, to $\mathbf{y}'_d$ assuming $\mathbf{x}'_v=\mathbf{0}$ for all $v\neq u$ is given by   
\begin{align}
&\mathbf{y}'_d|_{\mathbf{x}'_v=\mathbf{0},\forall v\neq u}\nonumber\\
&=\mathbf{F}_d\bigoplus_{v_{M-1}\in \Gamma^{u}_{in}(d)}\mathbf{H}_{v_{M-1},d}\mathbf{x}_{v_{M-1}}\nonumber\\
&=\mathbf{F}_d\bigoplus_{v_{M-1}\in \Gamma^{u}_{in}(d)}\mathbf{H}_{v_{M-1},d}\mathbf{F}_{v_{M-1}}\bigoplus_{v_{M-2}\in \Gamma^{u}_{in}(v_{M-1})}\mathbf{H}_{v_{M-2},v_{M-1}}\mathbf{F}_{v_{M-2}}\cdots\bigoplus_{v_{2}\in \Gamma^{u}_{in}(v_{3})}\mathbf{H}_{v_2,v_3}\mathbf{x}_{v_2}\nonumber\\
&=\mathbf{F}_d\bigoplus_{v_{M-1}\in \Gamma^{u}_{in}(d)}\mathbf{H}_{v_{M-1},d}\mathbf{F}_{v_{M-1}}\bigoplus_{v_{M-2}\in \Gamma^{u}_{in}(v_{M-1})}\mathbf{H}_{v_{M-2},v_{M-1}}\mathbf{F}_{v_{M-2}}\cdots\bigoplus_{v_{2}\in \Gamma^{u}_{in}(v_{3})}\mathbf{H}_{v_2,v_3}\mathbf{F}_{v_2}\mathbf{H}_{u,v_2}\mathbf{F}_u\mathbf{x}''_{u}.
\end{align}
Let us denote
\begin{equation}
\mathbf{H}_u:=\mathbf{F}_d\bigoplus_{v_{M-1}\in \Gamma^{u}_{in}(d)}\cdots \bigoplus_{v_{2}\in \Gamma^{u}_{in}(v_{3})}\mathbf{H}_{v_{M-1},d}\mathbf{F}_{v_{M-1}}\cdots\mathbf{H}_{v_2,v_3}\mathbf{F}_{v_2}\mathbf{H}_{u,v_2}\mathbf{F}_u,
\end{equation}
which is the $n\tau \times n\tau$ dimensional end-to-end channel matrix from $\mathbf{x}''_{u}$ to $\mathbf{y}'_d$.
Then $\mathbf{y}'_d$ can be represented as 
\begin{align} \label{eq:equi_ch}
\mathbf{y}'_d &=\bigoplus_{u\in V[1]}\mathbf{H}_u \mathbf{x}''_u\nonumber\\
&=\bigoplus_{i=1}^K \mathbf{H}_{t_i}\mathbf{x}''_{t_i}.
\end{align}
The following theorem and corollary show that if the size of the end-to-end channel matrix $\mathbf{H}_{t_i}$ is smaller than the corresponding minimum-cut value, then $\mathbf{H}_{t_i}$ becomes a full-rank matrix for all $i\in[1:K]$ with probability approaching one as $n$ increases. 

\begin{theorem}[Avestimehr--Diggavi--Tse \cite{Avestimehr:11}] \label{thm:ADT}
For any $i\in [1:K]$, let $\{\mathbf{x}''_{t_i}(w)\}_{w\in[1:2^{n \tau \log q}]}$ be a set of randomly chosen $2^{n\tau\log q}$ vectors in $\mathbb{F}^n_q$ and $\{\mathbf{y}'_d(w)\}_{w\in[1:2^{n \tau\log q}]}$ be the corresponding set of output vectors, i.e., $\mathbf{y}'_d(w) =\mathbf{H}_{t_i} \mathbf{x}''_{t_i}(w)$.
Suppose that the elements of $\mathbf{F}_v$ are  i.i.d. drawn uniformly from $\mathbb{F}_q$ for all $v\in V$.
Then there exists one-to-one correspondence between $\{\mathbf{x}''_{t_i}(w)\}_{w\in[1:2^{n\tau \log q}]}$ and $\{\mathbf{y}'_d(w)\}_{w\in[1:2^{n\tau\log q}]}$ with probability approaching one as $n$ increases, provided that 
\begin{equation}
\tau\leq \min _{\Omega\in \Lambda(\{i\})}\operatorname{rank}\left(H_{\Omega}(\{i\})\right).
\end{equation}
\end{theorem}
\begin{proof}
We refer to \cite[Theorem 4.1]{Avestimehr:11} for the proof.
\end{proof}

\begin{corollary} \label{co:rank_condition}
Suppose that the elements of $\mathbf{F}_v$ are chosen i.i.d. uniformly from $\mathbb{F}_q$ for all $v\in V$.
If $\tau\leq\min_{i\in [1:K]}\min _{\Omega\in \Lambda(\{i\})}\operatorname{rank}\left(H_{\Omega}(\{i\})\right)$, then
\begin{equation} \label{eq:rank_condition}
\operatorname{rank}(\mathbf{H}_{t_i})=n\tau
\end{equation}
for all $i\in [1:K]$ with probability approaching one as $n$ increases.
\end{corollary}
\begin{proof}
Assume that $\operatorname{rank}(\mathbf{H}_{t_i})<n\tau$.
Then, since the vector space spanned by $\mathbf{H}_{t_i}$ is $\mathbb{F}_q^{\operatorname{rank}(\mathbf{H}_{t_i})}$ space, it contains strictly less than $q^{n\tau}$ distinguishable vectors in $\mathbb{F}_q^{n\tau}$.
On the other hand, Theorem \ref{thm:ADT} shows that it is possible to have $2^{n\tau \log q}=q^{n\tau}$
distinguishable $\mathbf{y}'_d$'s with probability approaching one as $n$ increases, which 
contradicts the assumption.
In conclusion, from the union bound, \eqref{eq:rank_condition} holds for all $i\in [1:K]$ with probability approaching one as $n$ increases, which completes the proof.
\end{proof}

Based on Corollary \ref{co:rank_condition}, we set such that the elements of $\mathbf{F}_v$ are independently and uniformly chosen from $\mathbb{F}_q$ for all $v\in V$ and 
\begin{equation}
\tau=\min_{i\in [1:K]}\min _{\Omega\in \Lambda(\{i\})}\operatorname{rank}\left(H_{\Omega}(\{i\})\right),
\end{equation}
which guarantees the existence of $\mathbf{H}_{t_i}^{-1}$ for all $i\in [1:K]$ with probability approaching one as $n$ increases.
Then by setting $\mathbf{x}''_{t_i}=\mathbf{H}_{t_i}^{-1} \mathbf{x}'_{t_i}$, $\mathbf{x}'_{t_i}\in\mathbb{F}^{n\tau}_q$, for $n$ sufficiently large, we have 
\begin{equation}
\mathbf{y}'_d =\bigoplus_{i=1}^K \mathbf{x}'_{t_i}
\end{equation}
from \eqref{eq:equi_ch}. 
In conclusion, Lemma \ref{thm:transform} holds for layered networks.

\subsection{Arbitrary Networks} \label{subsec:arbitrary}
In this subsection, we prove Lemma \ref{thm:transform} for a general linear finite field network (not necessarily layered).
We can unfold the network $G$ over time to establish the corresponding layered network.
The underlying approach is similar to that proposed in \cite[Section V. B]{Avestimehr:11}. 
Let
\begin{align}
\bar{C}_{\min}&:=\min_{i\in [1:K]}\min _{\Omega\in \Lambda(\{i\})}\operatorname{rank}\left(H_{\Omega}(\{i\})\right)\log q
\end{align}
Define the $T$ time-steps unfolded network $G_{\sf TU}=(V_{\sf TU}, E_{\sf TU})$ as follows.
\begin{itemize}
\item The network has $T+2$ stages, numbered from $0$ to $T+1$.
\item Stage $0$ has  the senders $t_1[0]$ to $t_K[0]$, which are the senders, and stage $T+1$ has  node $d[T+1]$, which is the receiver.
\item Stage $j$ has all nodes $v\in V$ denoted by $v[j]$, where $j\in[1:T]$. These nodes will act as relay nodes.  
\item There are finite-capacity links with the rate of $T\bar{C}_{\min}$ between
\begin{itemize}
\item $(t_i[0],t_i[1])$ for all $i\in[1:K]$ and $(d[T],d[T+1])$.
\item $(v[j],v[j+1])$ for all $v\in V$ and $j\in[1:T]$.
\end{itemize}
\item Node $v[j]$ is connected to node $w[j+1]$ with the linear finite field channel of the original network $G$ for all $(v,w)\in E, v\neq w$.
\end{itemize}

The length-$n$ time-extended transmit signal of $v[j], j\in[0:T]$, is given by the pair of $\left(\mathbf{x}^{(1)}_{v[j]},\mathbf{x}^{(2)}_{v[j]}\right)$, where $\mathbf{x}^{(1)}_{v[j]}\in \mathbb{F}^{nT\bar{C}_{\min}/\log q}_q$ and $\mathbf{x}^{(2)}_{v[j]}\in \mathbb{F}^{n\alpha_{v[j]}}_q$.
The  length-$n$ time-extended received signal of $v[1]\in \Gamma_{out}(t_i[0])$, $i\in[1:K]$, is given by $\mathbf{y}^{(1)}_{v[1]}=\mathbf{x}^{(1)}_{v[0]}$.
The length-$n$ time-extended received signal of $v[j], j\in[2:T]$, is given by $\left(\mathbf{y}^{(1)}_{v[j]},\mathbf{y}^{(2)}_{v[j]}\right)$, where $\mathbf{y}^{(1)}_{v[j]}=\mathbf{x}^{(1)}_{v[j-1]}$ and $\mathbf{y}^{(2)}_{v[j]}=\bigoplus_{v[j-1]\in \Gamma(v[j])}\mathbf{H}_{v[j-1],v[j]}\mathbf{x}^{(2)}_{v[j-1]}$ (see the input--output relation \eqref{eq:input_output_field2}). 
The length-$n$ time-extended received signal of $v[T+1]$ is given by $\mathbf{y}^{(1)}_{v[T+1]}=\mathbf{x}^{(1)}_{v[T]}$.
For other unspecified received signals, they receive all-zero vectors. 
For a better understanding, Fig. \ref{figs:ex_unfold} illustrates an example of an $T$ time-steps unfolded network.

\begin{figure}[t!]
\begin{center}
\includegraphics[scale=1]{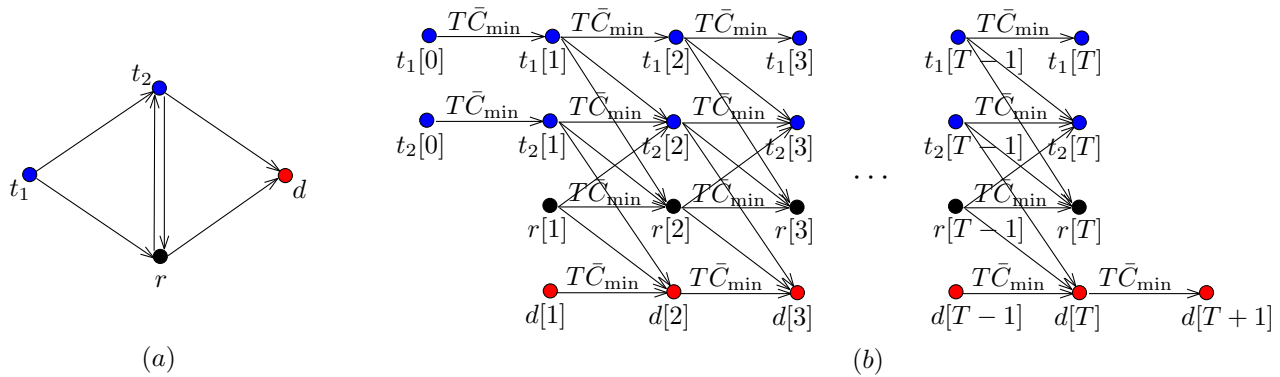}
\end{center}
\vspace{-0.15in}
\caption{Unfolded network example, where $(a)$ is the original network and $(b)$ is the corresponding unfolded network.}
\label{figs:ex_unfold}
\vspace{-0.1in}
\end{figure}

Since this unfolded network $G_{\sf TU}$ is layered, we can apply the same linear coding as in Appendix II. A to $G_{\sf TU}$. Specifically,
\begin{itemize}
\item (Sender Encoding) Node $v[0]$ transmits $\mathbf{x}^{(1)}_{v[0]}=\mathbf{F}^{(1)}_{v[0]}\mathbf{x}''_{v[0]}$ and $\mathbf{x}^{(2)}_{v[0]}=\mathbf{F}^{(2)}_{v[0]}\mathbf{x}''_{v[0]}$, where $\mathbf{F}^{(1)}_{v[0]}\in\mathbb{F}_q^{nT\bar{C}_{\min}/\log q \times n\tau}$, $\mathbf{F}^{(2)}_{v[0]}\in\mathbb{F}_q^{n\alpha_{v[0]} \times n\tau}$,  and $\mathbf{x}''_{v[0]}\in \mathbb{F}_q^{n\tau}$. We will specify $\tau$ later.
\item (Relay Encoding) Node $v[j]$, $j\in[1:T]$, transmits $\mathbf{x}^{(1)}_{v[j]}=\mathbf{F}^{(1)}_{v[j]}\left[{\mathbf{y}^{(1)}_{v[j]}}^T,{\mathbf{y}^{(2)}_{v[j]}}^T\right]^T$ and $\mathbf{x}^{(2)}_{v[j]}=\mathbf{F}^{(2)}_{v[j]}\left[{\mathbf{y}^{(1)}_{v[j]}}^T,{\mathbf{y}^{(2)}_{v[j]}}^T\right]^T$, where $\mathbf{F}^{(1)}_{v[j]}\in\mathbb{F}_q^{nT\bar{C}_{\min}/\log q \times n(T\bar{C}_{\min}/\log q+\beta_{v[j]})}$ and $\mathbf{F}^{(2)}_{v[j]}\in\mathbb{F}_q^{n\alpha_{v[j]} \times n(T\bar{C}_{\min}/\log q+\beta_{v[j]})}$.
\end{itemize} 
Then the receiver generates $\mathbf{y}'_{v[T+1]}=\mathbf{F}^{(1)}_{v[T+1]}\mathbf{y}^{(1)}_{v[T+1]}$, where $\mathbf{F}^{(1)}_{v[T+1]}\in\mathbb{F}_q^{n\tau\times nT\bar{C}_{\min}/\log q}$.

Similar to \eqref{eq:equi_ch}, the input--output from $\{\mathbf{x}''_{v[0]}\}_{v[0]\in V[0]}$ to $\mathbf{y}'_{v[T+1]}$ can be represented as 
\begin{equation} \label{equi_ch_arb}
\mathbf{y}'_{v[T+1]}=\bigoplus_{v[0]\in V[0]}\mathbf{H}_{v[0]} \mathbf{x}''_{v[0]},
\end{equation}
where $V[0]$ denotes the set of nodes at stage $0$, which is the set of senders, and $v[T+1]$ is the node at state $T+1$, which is the receiver.
Here $\mathbf{H}_{v[0]}$ is the end-to-end channel matrix from $\mathbf{x}''_{v[0]}$ to $\mathbf{y}'_{v[T+1]}$.

Now consider the minimum-cut value of $G_{\sf TU}$ with respect to a subset of nodes in $V[0]$.
In the same manner in Section \ref{subset:cut_set}, for $\Sigma\subseteq V[0]$,  we can define $G_{\sf TU}(\Sigma)$ and $\Lambda_{\sf TU}(\Sigma)$.
Then the minimum-cut value is given by
\begin{align} 
\bar{C}_{\sf TU}(\Sigma):=\min_{\Omega\in \Lambda_{\sf TU}(\Sigma)}\operatorname{rank}\left(H^{\sf TU}_{\Omega}(\Sigma)\right)\log q,
\end{align}
where $H^{\sf TU}_{\Omega}(\Sigma)$ denotes the transfer matrix associated with the cut $\Omega\in\Lambda_{\sf TU}(\Sigma)$ on $G_{\sf TU}(\Sigma)$.

Hence, from Theorem \ref{thm:ADT} and Corollary \ref{co:rank_condition}, by setting the elements of $\mathbf{F}^{(1)}_{v[j]}$ and $\mathbf{F}^{(2)}_{v[j]}$ i.i.d. drawn uniformly from $\mathbb{F}_q$, we can guarantee that
\begin{align}
\operatorname{rank}(\mathbf{H}_{v[0]})=n\tau
\end{align}
for all $v[0]\in V[0]$ with probability approaching one as $n$ increases if 
\begin{equation} \label{eq:tau_arb}
\tau\leq \min_{v[0]\in V[0]}\frac{\bar{C}_{\sf TU}(\{v[0]\})}{\log q}.
\end{equation}

For $v[0]\in V[0]$, the minimum-cut value is lower bounded by
\begin{align} \label{eq:min_cut_arb}
\bar{C}_{\sf TU}(\{v[0]\})&=\min_{\Omega\in \Lambda_{\sf TU}(\{v[0]\})}\operatorname{rank}\left(H^{\sf TU}_{\Omega}(\{v[0]\})\right)\log q\nonumber\\
&\geq (T-|V|)\min_{\Omega\in\Lambda(\{i\})}\operatorname{rank} (H_{\Omega}(\{i\}))\log q,
\end{align}
where $t_i=v[0]$. The inequality follows from the same analysis in \cite[Lemma 5.2]{Avestimehr:11}.

From \eqref{eq:tau_arb} and \eqref{eq:min_cut_arb}, we set 
\begin{equation}
\tau= (T-|V|)\min_{i\in[1:K]}\min_{\Omega\in\Lambda(\{i\})}\operatorname{rank} (H_{\Omega}(\{i\}))
\end{equation}
and the elements of $\mathbf{F}^{(1)}_{v[j]}$ and $\mathbf{F}^{(2)}_{v[j]}$ i.i.d. drawn uniformly from $\mathbb{F}_q$. 
This guarantees the existence of $\mathbf{H}_{v[0]}^{-1}$ for all $v[0]\in V[0]$ with probability approaching one as $n$ increases.
Therefore from \eqref{equi_ch_arb}, setting $\mathbf{x}''_{v[0]}=\mathbf{H}^{-1}_{v[0]}\mathbf{x}'_{v[0]}$, $\mathbf{x}'_{v[0]}\in\mathbb{F}^{n\tau}_q$, provides
\begin{equation} \label{eq:final_in_out}
\mathbf{y}^{(1)}_{v[T+1]}=\bigoplus_{v[0]\in V[0]} \mathbf{x}'_{v[0]}.
\end{equation}
Finally, since any coding scheme for the $T$ time-steps unfolded network $G_{\sf TU}$ can be performed in the original network $G$ using $nT$ time slots, see the argument in \cite[Lemma 5.1]{Avestimehr:11}, we have 
\begin{equation} \label{eq:final_in_out2}
\mathbf{y}'_d=\bigoplus_{i=1}^K \mathbf{x}'_{t_i},
\end{equation}
where $\mathbf{x}'_{t_i}\in\mathbb{F}^{n\tau}_q$ for all $i\in[1:K]$ using $nT$ time slots for the original network $G$.
Here, we simply rewrite $\{\mathbf{x}'_{t_i}\}_{i\in [1:K]}=\{\mathbf{x}'_{v[0]}\}_{v[0]\in V[0]}$ and $\mathbf{y}'_d=\mathbf{y}^{(1)}_{v[T+1]}$ from \eqref{eq:final_in_out} since $V[0]$ is the set of senders and $v[T+1]$ is the receiver.
Then using $n$ time slots, we  have \eqref{eq:final_in_out2} with $\mathbf{x}'_{t_i}\in\mathbb{F}^{n\tau/T}_q$.
From the fact that $\lim_{T\to\infty}\frac{\tau}{T}\to \min_{i\in[1:K]}\min_{\Omega\in\Lambda(\{i\})}\operatorname{rank} (H_{\Omega}(\{i\}))$, Lemma \ref{thm:transform} holds for any arbitrary networks.


\end{document}